\documentclass[12pt]{article}

\renewcommand{\arraystretch}{1}

\usepackage{amsmath}
\usepackage{amssymb}
\usepackage{amsthm}
\usepackage{hyperref}
\usepackage[pdftex]{graphicx}
\usepackage[english]{babel}
\usepackage{array}
\usepackage{fancyhdr}
\usepackage{enumerate}

\newtheorem{theorem}{Theorem}[section]
\newtheorem{corollary}[theorem]{Corollary}
\newtheorem{lemma}[theorem]{Lemma}

\newtheorem{assumption}[theorem]{Assumption}

\newtheorem{definition}[theorem]{Definition}

\newtheorem{example}[theorem]{Example}

\newtheorem{remark}[theorem]{Remark}
\newenvironment{rem}{\begin{remark}\rm}{\end{remark}}
\newtheorem{tab}{Table}

\newcommand{\argmax}{\operatornamewithlimits{argmax}}
\def\EE{{\bf E}}
\newcommand{\eps}{\varepsilon}

\def\p{{\mathbf p}}

\title{Asymptotics of Input-Constrained Erasure Channel Capacity~\footnote{A preliminary version of this work has been presented in IEEE ISIT 2014.}}

\author{\begin{tabular}{cc}
Yonglong Li&Guangyue Han\\
The University of Hong Kong&The University of Hong Kong\\
{\em email:} yonglong@hku.hk&{\em email:} ghan@hku.hk\\
\end{tabular}}

\date{{\normalsize \today}}

\begin{document}\maketitle\thispagestyle{empty}

\begin{abstract}
In this paper, we examine an input-constrained erasure channel and we characterize the asymptotics of its capacity when the erasure rate is low. More specifically, for a general memoryless erasure channel with its input supported on an irreducible finite-type constraint, we derive partial asymptotics of its capacity, using some series expansion type formulas of its mutual information rate; and for a binary erasure channel with its first-order Markovian input supported on the $(1, \infty)$-RLL constraint, based on the concavity of its mutual information rate with respect to some parameterization of the input, we numerically evaluate its first-order Markov capacity and further derive its full asymptotics. The asymptotics obtained in this paper, when compared with the recently derived feedback capacity for a binary erasure channel with the same input constraint, enable us to draw the conclusion that feedback may increase the capacity of an input-constrained channel, even if the channel is memoryless.
\end{abstract}


\section{Introduction} \label{introduction-section}

The primary concern of this paper is the erasure channel, which is a common digital communication channel model that plays a fundamental role in coding and information theory. Throughout the paper, we assume that time is discrete and indexed by the integers. At time $n$, the erasure channel of interest can be described by the following equation:
\begin{equation} \label{mec}
Y_n=X_n E_n,
\end{equation}
where the channel input $\{X_n\}$, supported on an irreducible finite-type constraint $\mathcal{S}$, is a stationary process taking values from the input alphabet $\mathcal{X}=\{1,2,\cdots, K\}$ , and the erasure process $\{E_n\}$, independent of $\{X_n\}$, is a binary stationary and ergodic process with {\em erasure rate} $\eps \triangleq P(E_1=0)$, and $\{Y_n\}$ is the channel output process over the output alphabet $\mathcal{Y}=\{0, 1, \cdots, K\}$. The word ``erasure'' as in the name of our channel naturally arises if a ``$0$'' is interpreted as an erasure at the receiving end of the channel; so, at time $n$, the channel output $Y_n$ is nothing but the channel input $X_n$ if $E_n=1$, but an erasure if $E_n=0$.

Let $\mathcal{X}^*$ denote the set of all the finite length words over $\mathcal{X}$. Let $\mathcal{F}$ be a finite subset of $\mathcal{X}^*$, and let $\mathcal{S}$ be the {\em finite-type constraint} with respect to $\mathcal{F}$, which is a subset of $\mathcal{X}^*$ consisting of all the finite length words over $\mathcal{X}$, each of which does not contain any element in $\mathcal{F}$ as a contiguous subsequence (or, roughly, elements in $\mathcal{F}$ are ``forbidden'' in $\mathcal{S}$). The most well known example is the $(d,k)$-run-length-limited (RLL) constraint over the alphabet $\{1, 2\}$, which forbids any sequence with fewer than $d$ or more than $k$ consecutive $1$'s in between two successive $2$'s; in particular, a prominent example is the $(1,\infty)$-RLL constraint, a widely used constraint in magnetic recording and data storage; see~\cite{symbolicmarcuslind, mrs98}. For the $(d, k)$-RLL constraint with $k < \infty$, a forbidden set $\mathcal{F}$ is
$$
\mathcal{F}=\{2\underbrace{1\cdots1}_l2: 0 \leq l < d \} \cup \{\underbrace{0\cdots0}_{k+1}\}.
$$
When $k=\infty$, one can choose $\mathcal{F}$ to be
$$
\mathcal{F}=\{2\underbrace{1\cdots1}_l2: 0 \leq l < d \};
$$
in particular when $d=1, k=\infty$, $\mathcal{F}$ can be chosen to be $\{22\}$. The {\em length} of $\mathcal{F}$ is defined to be that of the longest words in $\mathcal{F}$. Generally speaking, there may be many such $\mathcal{F}$'s with different lengths that give rise to the same constraint $\mathcal{S}$; the length of the shortest such $\mathcal{F}$'s minus $1$ gives the {\em topological order} of $\mathcal{S}$. For example, the topological order of the $(1, \infty)$-RLL constraint, whose shortest $\mathcal{F}$ proves to be $\{22\}$, is $1$. A finite-type constraint $\mathcal{S}$ is said to be {\em irreducible} if for any $u, v \in \mathcal{S}$, there is a $w \in \mathcal{S}$ such that $uwv \in \mathcal{S}$.

As mentioned before, the input process $X$ of our channel (\ref{mec}) is assumed to be {\em supported} on an irreducible finite-type constraint $\mathcal{S}$, namely, $\mathcal{A}({X})\subseteq \mathcal{S}$, where
$$
\mathcal{A}({X}) \triangleq \{x_{i}^{j} \in \mathcal{X}^*:p_{X}(x_{i}^{j})>0 \}.
$$
The capacity of the channel (\ref{mec}), denoted by $C(\mathcal{S},\eps)$, can be computed by
\begin{eqnarray*}
C(\mathcal{S},\eps)=\sup_{\mathcal{A}({X})\subseteq \mathcal{S}} I(X;Y),
\end{eqnarray*}
where the supremum is taken over all stationary processes $X$ supported on $\mathcal{S}$. Here, we note that input-constraints~\cite{ZehaviWolf88} are widely used in various real-life applications such as magnetic and optical recording~\cite{mrs98} and communications over band-limited channels with inter-symbol interference~\cite{fo72}. Particularly, we will pay special attention in this paper to a binary erasure channel with erasure rate $\eps$ (BEC($\eps$)) with the input supported on the $(1, \infty)$-RLL constraint, denoted by $\mathcal{S}_0$ throughout the paper.

When there is no constraint imposed on the input process $X$, that is, $\mathcal{S}=\mathcal{X}^*$, it is well known that $C(\mathcal{S}, \eps)=(1-\eps) \log K$; see Theorem~\ref{fbnot}. When $\eps=0$, that is, when the channel is perfect with no erasures, $C(\mathcal{S}, \eps)$ proves to be the {\em noiseless capacity} of the constraint $\mathcal{S}$, which can be achieved by a unique $m$-th order Markov chain $\hat{X}$ with $\mathcal{A}(\hat{X})=\mathcal{S}$~\cite{parryintrinsicmarkovchains}. On the other hand, other than these two above-mentioned ``degenerated'' cases, ``explicit'' analytic formulas of capacity for ``non-degenerated'' cases have remained evasive, and the problem of analytically characterizing the noisy constrained capacity is widely believed to be intractable.

The problem of numerically computing the capacity $C(\mathcal{S}, \eps)$ seems to be as challenging: the computation of the capacity of a general channel with memory or input constraints is notoriously difficult and has been open for decades; and the fact that our erasure channel is only a special class of such ones does not appear to make the problem easier. Here, we note that for a discrete memoryless channel, Shannon gave a closed-form formula of the capacity in his celebrated paper~\cite{Shannon}, and Blahut~\cite{Blahut} and Arimoto~\cite{Arimoto}, independently proposed an algorithm which can efficiently compute the capacity and the capacity-achieving distribution simultaneously. However, unlike the discrete memoryless channels, the capacity of a channel with memory or input constraints in general admits no single-letter characterization and very little is known about the efficient computation of the channel capacity. To date, most known results in this regard have been in the forms of numerically computed bounds: for instance, numerically computed lower bounds by Arnold and Loeliger~\cite{arnoldinformationrate}, A.~Kavcic~\cite{kavcic2001}, Pfister, Soriaga and Siegel~\cite{pfister2001}, Vontobel and Arnold~\cite{vontobel2001}.

One of the most effective strategies to compute the capacity of channels with memory or input constraints is the so-called {\em Markov approximation} scheme. The idea is that instead of maximizing the mutual information rate over all stationary processes, one can maximize the mutual information rate over all $m$-th order Markov processes to obtain the $m$-th order Markov capacity. Under suitable assumptions (see, e.g.,~\cite{chensiegel}), when $m$ tends to infinity, the corresponding sequence of Markov capacities will tend to the channel capacity. For our erasure channel, the {\em $m$-th order Markov capacity} is defined as
\begin{eqnarray*}
C^{(m)}(\mathcal{S},\eps)=\sup I(X;Y),
\end{eqnarray*}
where the supremum is taken over all $m$-th order Markov chains supported on $\mathcal{S}$.

The main contributions of this work are the characterization of the asymptotics of the above-mentioned input-constrained erasure channel capacity. Of great relevance to this work are results by Han and Marcus~\cite{asymptotics-binary}, Jacquet and Szpankowski~\cite{noisyconstrainedcapacityforbsc}, which have characterized asymptotics of the capacity of the a binary symmetric channel with crossover probability $\eps$ (BSC($\eps$)) with the input supported on the $(1, \infty)$-RLL constraint. The approach in the above-mentioned work is to obtain the asymptotics of the mutual information rate first, and then apply some bounding argument to obtain that of the capacity. The approach in this work roughly follows the same strategy, however, as elaborated below, our approach differs from theirs to a great extent in terms of technical implementations.

Throughout the paper, we use the logarithm with base $e$ in the proofs and we use the logarithms with base $2$ in the numerical computations of the channel capacity. Below is a brief account of our results and methodology employed in this work.

The starting point of our approach is Lemma~\ref{nformula} in Section~\ref{mutual-information-rate-section}, a key lemma that expresses the conditional entropy $H(Y_0|Y_{-n}^{-1})$ in a form that is particularly effective for analyzing asymptotics of $C(\mathcal{S}, \eps)$ when $\eps$ is close to $0$. As elaborated in Theorem~\ref{wolf'sconjecture}, Lemma~\ref{nformula} naturally gives a lower and upper bound on $C(\mathcal{S}, \eps)$, where the lower bound gives a counterpart result of Wolf's conjecture for a BEC($\eps$). Moreover, when applied to the case when $X$ is a Markov chain, Lemma~\ref{nformula} yields some explicit series expansion type formulas in Theorem~\ref{entropyformula} and Corollary~\ref{memorylessec}, which aptly pave the way for characterizing the asymptotics of the input-constrained erasure channel capacity. Here we remark that the method in~\cite{asymptotics-binary, noisyconstrainedcapacityforbsc} have been further developed for more general families of memory channels in~\cite{asymptotics, hm09b} via examining the contractiveness of an associated random dynamical system~\cite{davidblackwell1957}. However, the methodology to derive asymptotics of the mutual information rate in this work capitalizes on certain characteristics that are in a sense unique to erasure channels.

In Section~\ref{general-asymptotics}, we consider a memoryless erasure channel with the input supported on an irreducible finite-type constraint, and in Theorem~\ref{asyerasurec}, we derive partial asymptotics of its capacity $C(\mathcal{S}, \eps)$ in the vicinity of $\eps=0$ where $C(\mathcal{S}, \eps)$ is written as the sum of a constant term, a linear term in $\eps$ and an $O(\eps^2)$-term. The lower bound part in the proof of this theorem follows from an easy application of Theorem~\ref{entropyformula}, and the upper bound part hings on an adapted argument in~\cite{asymptotics-binary}.

In Section~\ref{binary-asymptotics}, we consider a BEC($\eps$) with the input being a first-order Markov process supported on the $(1, \infty)$-RLL constraint $\mathcal{S}_0$. Within this special setup, we show in Theorem~\ref{concavitybec} that the $I(X; Y)$ is strictly concave with respect to some parameterization of $X$. And in Section~\ref{sub-2}, we numerically evaluate $C^{(1)}(\mathcal{S}_0, \eps)$ and the corresponding capacity-achieving distribution using the randomized algorithm proposed in~\cite{randomapproachhan} which proves to be convergent given the concavity of $I(X; Y)$. Moreover, the concavity of $I(X; Y)$ guarantees the uniqueness of the capacity achieving distribution, based on which we derive full asymptotics of the above input-constrained BEC($\eps$) around $\eps=0$ in Theorem~\ref{foc}, where $C^{(1)}(\mathcal{S}, \eps)$ is expressed as an infinite sum of all $O(\eps^k)$-terms.

In Section~\ref{feedback-section}, we turn to the scenarios when there might be feedback in our erasure channel. We first prove in Theorem~\ref{fbnot} that when there is no input constraint, the feedback does not increase the capacity of the erasure channel even with the presence of the channel memory. When the input constraint is not trivial, however, we show in Theorem~\ref{yonglong-feedback} that feedback does increase the capacity using the example of a BEC($\eps$) with the ($1, \infty$)-RLL input constraint, and so feedback may increase the capacity of input-constrained erasure channels even if there is no channel memory. The results obtained in this section suggest the intricacy of the interplay between feedback, memory and input constraints.

\section{A Key Lemma and Its Applications} \label{mutual-information-rate-section}

In this section, we focus on the mutual information of the erasure channel (\ref{mec}) introduced in Section~\ref{introduction-section}. The starting point of our approach is the following key lemma, which is particularly effective for analysis of input-constrained erasure channels.
\begin{lemma}\label{nformula}
For any $n \geq 1$, we have
\begin{equation}\label{nformula1}
H(Y_0|Y_{-n}^{-1})=H(E_0|E_{-n}^{-1})+\sum_{D \subseteq [-n,-1]}H(X_{0}|X_{D})P(E_0=1,E_{D}=1,E_{D^c}=0),
\end{equation}
where $[-n,-1] \triangleq \{-n,\cdots,-1\}$.
\end{lemma}

\begin{proof}
Note that
\begin{eqnarray*}
H(Y_0|Y_{-n}^{-1})&=&-\sum_{y_{-n}^{0}}p(y_{-n}^{0})\log p(y_{0}|y_{-n}^{-1})\\
&=&T_1(n)+T_2(n),
\end{eqnarray*}
where
$$
T_1(n)=-\sum_{y_{-n}^{-1},y_{0}=0}p(y_{-n}^{0})\log p(y_{0}|y_{-n}^{-1})
\qquad
\mbox{and}
\qquad
T_2(n)=-\sum_{y_{-n}^{-1},y_{0}\not=0}p(y_{-n}^{0})\log p(y_{0}|y_{-n}^{-1}).
$$
From the independence of $\{X_n\}$ and $\{E_n\}$, it follows that
\begin{align}\label{pformula}
p(y_{i}^{j})&=\sum_{x_{i}^{j}:\ x_{k}=y_{k}\  \mbox{\scriptsize for } k \in \mathcal{I}(y_{i}^{j})}p_{X}(x_{i}^{j})P(E_{\mathcal{I}(y_{i}^{j})}=1,E_{\bar{\mathcal{I}}(y_{i}^{j})}=0)\nonumber\\
&=p_{X}\left(y_{\mathcal{I}(y_{i}^{j})}\right)P(E_{\mathcal{I}(y_{i}^{j})}=1,E_{\bar{\mathcal{I}}(y_{i}^{j})}=0).
\end{align}
Here and throughout the paper, let $\mathcal{Y}^*$ be the set of all finite length words over $\mathcal{Y}$ and we define, for any $y_i^j \in \mathcal{Y}^*$,
$$
\mathcal{I}(y_{i}^{j})=\{k:i\le k\le j,y_{k}\not=0\}, \quad \bar{\mathcal{I}}(y_{i}^{j})=\{k:i\le k\le j,y_{k}=0\}
$$
and
$$
y_{\mathcal{I}(y_{i}^{j})}=\{y_{k}:k\in \mathcal{I}(y_{i}^{j})\}.
$$
For $y_{0}\not=0$,
\begin{eqnarray*}
p(y_0|y_{-n}^{-1})&=&\frac{p(y_{-n}^{0})}{p(y_{-n}^{-1})}\\
&=&\frac{p_{X}\left(y_{\mathcal{I}(y_{-n}^{0})}\right)P(E_{\mathcal{I}(y_{-n}^{0})}=1,E_{\bar{\mathcal{I}}(y_{-n}^{0})}=0)}{p_{X}\left(y_{\mathcal{I}(y_{-n}^{-1})}\right)P(E_{\mathcal{I}(y_{-n}^{-1})}=1,E_{\bar{\mathcal{I}}(y_{-n}^{-1})}=0)}\\
&\stackrel{(a)}{=}&p_{X}\left(y_{0}|y_{\mathcal{I}(y_{-n}^{-1})}\right)P(E_0=1|E_{\mathcal{I}(y_{-n}^{-1})}=1,E_{\bar{\mathcal{I}}(y_{-n}^{-1})}=0),
\end{eqnarray*}
where $(a)$ follows from the fact that $\bar{\mathcal{I}}(y_{-n}^{0})=\bar{\mathcal{I}}(y_{-n}^{-1})$. Similarly, for $y_0=0$,
\begin{eqnarray*}
p(y_0|y_{-n}^{-1})&=&\frac{p_{X}\left(y_{\mathcal{I}(y_{-n}^{0})}\right)P(E_{\mathcal{I}(y_{-n}^{0})}=1,E_{\bar{\mathcal{I}}(y_{-n}^{0})}=0)}{p_{X}\left(y_{\mathcal{I}(y_{-n}^{-1})}\right)P(E_{\mathcal{I}(y_{-n}^{-1})}=1,E_{\bar{\mathcal{I}}(y_{-n}^{-1})}=0)}\\
&\stackrel{(a)}{=}&P(E_0=1|E_{\mathcal{I}(y_{-n}^{-1})}=1,E_{\bar{\mathcal{I}}(y_{-n}^{-1})}=0),
\end{eqnarray*}
where $(a)$ follows from the fact that $\mathcal{I}(y_{-n}^{0})=\mathcal{I}(y_{-n}^{-1})$. Therefore,
\begin{align}
T_1(n)&=-\sum_{y_{-n}^{-1},y_{0}=0}p(y_{-n}^{0})\log p(y_{0}|y_{-n}^{-1})\notag\\
&=-\sum_{y_{-n}^{-1},y_{0}=0}p(y_{-n}^{0})\log P(E_0=0|E_{\mathcal{I}(y_{-n}^{-1})}=1,E_{\bar{\mathcal{I}}(y_{-n}^{-1})}=0)\notag\\
&=-\sum_{D \subseteq [-n,-1]}\sum_{y_{-n}^0:\mathcal{I}(y_{-n}^{-1})=D,y_{0}=0}p(y_{-n}^{0})\log P(E_0=0|E_{D})=1,E_{D^c}=0)\notag\\
&\stackrel{(a)}{=}-\sum_{D \subseteq [-n,-1]}P(E_0=0,E_{D}=1,E_{D^c}=0)\log P(E_0=0|E_{D}=1,E_{D^c}=0)\label{0erasure},
\end{align}
where $(a)$ follows from the fact that for any given $D \subseteq [-n, -1]$,
$$
\sum_{y_{-n}^0: \mathcal{I}(y_{-n}^{-1})=D,y_{0}=0}p(y_{-n}^{0})=P(E_0=0,E_{D}=1,E_{D^c}=0).
$$
Also, we have
\begin{align}
T_2(n)&=-\sum_{y_{-n}^{-1},y_{0}\not=0}p(y_{-n}^{0})\log p(y_{0}|y_{-n}^{-1})\notag\\
&=-\sum_{y_{-n}^{-1},y_{0}\not=0}p(y_{-n}^{0})\log p_{X}\left(y_{0}|y_{\mathcal{I}(y_{-n}^{-1})}\right)P(E_0=1|E_{\mathcal{I}(y_{-l}^{-1})}=1,E_{\bar{\mathcal{I}}(y_{-l}^{-1})}=0)\notag\\
&=T_3(n)-\sum_{y_{-n}^{-1},y_{0}\not=0}p(y_{-n}^{0})\log P(E_0=1|E_{\mathcal{I}(y_{-n}^{-1})}=1,E_{\bar{\mathcal{I}}(y_{-n}^{-1})}=0)\notag\\
&\stackrel{(a)}{=}T_3(n)-\sum_{D \subseteq [-n,-1]}P(E_0=1,E_{D}=1,E_{D^c}=0)\log P(E_0=1|E_{D}=1,E_{D^c}=0)\label{1erasure},
\end{align}
where $(a)$ follows from a similar argument as in the proof of~(\ref{0erasure}) and
$$
T_3(n)=-\sum_{y_{-n}^{-1},y_{0}\not=0}p(y_{-n}^{0})\log p_{X}(y_{0}|y_{\mathcal{I}(y_{-n}^{-1})}).
$$
From~(\ref{pformula}), it then follows that
\begin{align}
T_3(n)&=-\sum_{y_{-n}^{-1},y_{0}\not=0}p(y_{-n}^{0})\log p_{X}(y_{0}|y_{\mathcal{I}(y_{-n}^{-1})})\notag\\
&=-\sum_{D \subseteq [-n, -1]}\sum_{y_{-n}^{0}: \mathcal{I}(y_{-n}^{0})=D\cup\{0\}}p_{X}(y_{D},y_0)P(E_0=1,E_{D}=1,E_{D^c}=0)\log p_{X}(y_{0}|y_{D})\notag\\
&=\sum_{D \subseteq [-n, -1]}H(X_0|X_{D})P(E_0=1,E_{D}=1,E_{D^c}=0).\label{eerasure}
\end{align}
The desired formula for $H(Y_0|Y_{-n}^{-1})$ then follows from (\ref{0erasure}), (\ref{1erasure}) and (\ref{eerasure}).
\end{proof}

One of the immediate applications of Lemma~\ref{nformula} is the following lower and upper bounds on $C(\mathcal{S}, \eps)$.
\begin{theorem}\label{wolf'sconjecture}
$$
(1-\eps) C(\mathcal{S},0) \leq C(\mathcal{S},\eps) \leq (1-\eps) \log K.
$$
\end{theorem}

\begin{proof}
For the upper bound, it follows from Lemma~\ref{nformula} that
\begin{align}
I(X;Y)&=\lim_{n\to \infty}(H(Y_0|Y_{-n}^{-1})-H(Y_0|Y_{-n}^{-1},X_{-n}^0))\notag\\
&=\lim_{n\to \infty}\sum_{D \subseteq [-n, -1]}H(X_{0}|X_{D})P(E_0=1,E_{D}=1,E_{D^c}=0)\notag\\
&\stackrel{(a)}{\leq} \lim_{n\to \infty}\sum_{D \subseteq [-n,-1]}H(X_{0}) P(E_0=1,E_{D}=1,E_{D^c}=0)\notag\\
&\leq \lim_{n\to \infty}\sum_{D \subseteq [-n,-1]} P(E_0=1,E_{D}=1,E_{D^c}=0) \log K \notag\\
&=P(E_0=1) \log K \notag\\
&=(1-\eps) \log K, \notag
\end{align}
where we have used the fact that conditioning reduces entropy for $(a)$.

Assume $\mathcal{S}$ is of topological order $m$, and let $\hat{X}$ be the $m$-order Markov chain that achieves the noiseless capacity $C(\mathcal{S}, 0)$ of the constraint $\mathcal{\mathcal{S}}$. Again, it follows from Lemma~\ref{nformula} that
\begin{align}
I(\hat{X};Y)&=\lim_{n\to \infty}(H(Y_0|Y_{-n}^{-1})-H(Y_0|Y_{-n}^{-1},\hat{X}_{-n}^0))\notag\\
&=\lim_{n\to \infty}\sum_{D \subseteq [-n, -1]}H(\hat{X}_{0}|\hat{X}_{D})P(E_0=1,E_{D}=1,E_{D^c}=0)\notag\\
&\ge \lim_{n\to \infty}\sum_{D \subseteq [-n, -1]}H(\hat{X}_{0}|\hat{X}_{-m}^{-1},\hat{X}_{D})P(E_0=1,E_{D}=1,E_{D^c}=0)\notag\\
&\stackrel{(a)}{=} \lim_{n\to \infty}\sum_{D \subseteq [-n, -1]}H(\hat{X}_{0}|\hat{X}_{-m}^{-1})P(E_0=1,E_{D}=1,E_{D^c}=0)\notag\\
&=P(E_0=1) H(\hat{X}_{0}|\hat{X}_{-m}^{-1}) \notag\\
&=(1-\eps) C(\mathcal{S}, 0), \notag
\end{align}
where we have used the fact that $\{\hat{X}_n\}$ is an $m$-th order Markov chain for $(a)$.
\end{proof}

\begin{rem}
The upper bound part of Theorem~\ref{wolf'sconjecture} also follows from the well-known fact that (see Theorem~\ref{fbnot})
$$C(\mathcal{X}^*, \eps)=(1-\eps) \log K$$
and for any $\mathcal{S}$,
$$
C(\mathcal{S}, \eps) \leq C(\mathcal{X}^*, \eps),
$$
which is obviously true.

Let $C'(\mathcal{S},\eps)$ denote the capacity of a BSC($\eps$) with the $(d,k)$-RLL constraint. In~\cite{wolf'sconjecture} Wolf posed the following conjecture on $C'(\mathcal{S}, \eps)$:
$$
C'(\mathcal{S},\eps)\ge C'(\mathcal{S},0)(1-H(\eps)),
$$
where $H(\eps) \triangleq -\eps \log \eps- (1-\eps) \log (1-\eps)$. A weaker form of this bound has been established in~\cite{pkumar1992} by counting the possible subcodes satisfying the $(d,k)$-RLL constraint in some linear coding scheme, but the conjecture for the general case still remains open.

It is well known that $1-H(\eps)$ is the capacity of a BSC($\eps$) without any input constraint, and $1-\eps$ is the capacity of a BEC($\eps$) without any input constraint. So, for an input-constrained BEC($\eps$), the lower bound part of Theorem~\ref{wolf'sconjecture} gives a counterpart result of Wolf's conjecture.
\end{rem}

When applied to the channel with a Markovian input, Lemma~\ref{nformula} gives a relatively explicit series expansion type formula for the mutual information rate of (\ref{mec}).
\begin{theorem}\label{entropyformula}
Assume $\{X_n\}$ is an $m$-th order input Markov chain. Then,
\begin{eqnarray}\label{mutualinformationrate}
I(X;Y)&=&\sum_{k=0}^{\infty}\sum_{t=0}^{b(k-1,m)}\sum_{\{i_{1}^{t}\}\in B_{2}(k-1,t)}H(X_{0}|X_{i_{1}^{t}},X_{-k-m}^{-k-1})P(E_{A(k,i_1^t)}=1,E_{\bar{A}(k,i_1^t)}=0),
\end{eqnarray}
where
$A(k,i_1^t)=\{-k-m,\cdots,-k-1,i_1^t,0\}\,\mbox{and}\,\, \bar{A}(k,i_1^t)=\{-k-m,\cdots,0\}- A(k,i_1^t)$ and
$B_{2}(n,u)=\{\{i_1,\cdots,i_{u}\} \subseteq [-n, -1]: \mbox{for all $j=1,\cdots,u$}, \{i_{j},i_{j}+1,\cdots,i_{j}+m\} \not\subseteq \{ i_1,\cdots,i_{u}\}\}$ and $b(k-1,m)=(m-1)\left\lfloor \frac{k-1}{m}\right\rfloor+R(k-1)$, here $R(k-1)$ denotes the remainder of $k-1$ divided by $m$.
\end{theorem}

\begin{proof}
Note that
\begin{align*}
H(Y_0|X_{-n}^{0},Y_{-n}^{-1}) & =H(X_0E_0|X_{-n}^{0},E_{-n}^{-1},Y_{-n}^{-1})\\
&\stackrel{(a)}{=}H(E_0|E_{-n}^{-1}),
\end{align*}
where $(a)$ follows from the independence of $\{X_n\}$ and $\{E_{n}\}$. From Lemma~\ref{nformula}, it then follows that
\begin{align}\label{iformula}
I(X;Y)&=\lim_{n\to\infty}(H(Y_0|Y_{-n}^{-1})-H(Y_0|X_{-n}^{0},Y_{-n}^{-1}))\notag\\
&=\lim_{n\to\infty}\sum_{D \subseteq [-n, -1]}H(X_0|X_{D})P(E_0=1,E_{D}=1,E_{D^c}=0).
\end{align}
Now, letting
$$
B(n,u)=\{D \subseteq [-n, -1]:|D|=u\} \mbox{ and }
B_1(n,u)=B(n,u)-B_{2}(n,u),
$$
we deduce that, for $\{i_1^t\}\in B_2(k-1,t)$
{\small\begin{align*}
\sum_{D \subseteq[-n,-k-m-1]}P(E_{A(k,i_1^t)}=1,E_D=1,E_{\bar{A}(k,i_1^t)}=0,E_{[-n,-k-m-1]-D}=0)&=P(E_{A(k,i_1^t)}=1,E_{\bar{A}(k,i_1^t)}=0).
\end{align*}}
and
\begin{align*}
&\hspace{-1cm} \sum_{k=m}^{n}\ \sum_{\{i_1,\dots, i_k\}\in B_{1}(n,k)}H(X_0|X_{i_1^k})P(E_0=1,E_{i_1^k}=1,E_{\bar{i}_1^k}=0)\\
&=\sum_{k=0}^{n-m+1}\sum_{t=0}^{b(k-1,m)}\sum_{\{i_1^t\}\in B(k-1,t)}\sum_{D\subseteq [-n,-k-m-1]}\left\{H(X_0|X_{A(k,i_1^t)},X_D)\right.\\
&\hspace{4.5mm}\times \left.P(E_{A(k,i_1^t)}=1,E_D=1,E_{\bar{A}(k,i_1^t)}=0,E_{[-n,-k-m-1]-D}=0)\right\}\\
&\stackrel{(a)}{=}\sum_{k=0}^{n-m+1}\sum_{t=0}^{b(k-1,m)}\sum_{\{i_1^t\}\in B(k-1,t)}\sum_{D \subseteq [-n,-k-m-1]}\left\{H(X_0|X_{A(k,i_1^t)})\right.\\
&\hspace{4.5mm}\times\left.P(E_{A(k,i_1^t)}=1,E_D=1,E_{\bar{A}(k,i_1^t)}=0,E_{[-n,-k-m-1]-D}=0)\right\}\\
&=\sum_{k=0}^{n-m+1}\sum_{t=0}^{b(k-1,m)}\sum_{\{i_1^t\}\in B(k-1,t)}H(X_0|X_{A(k,i_1^t)})P(E_{A(k,i_1^t)}=1,E_{\bar{A}(k,i_1^t)}=0),
\end{align*}
where $(a)$ follows from the fact that $\{X_n\}$ is an $m$-th order Markov chain.
Then it follows that
\begin{align}
&\hspace{-1cm}\sum_{D \subseteq [-n, -1]}H(X_0|X_{D})P(E_0=1,E_{D}=1,E_{D^c}=0)\notag\\
&=\sum_{k=0}^{n}\ \sum_{\{i_1,\dots, i_k\}\in B(n,k)}H(X_0|X_{i_1^k})P(E_0=1,E_{i_1^k}=1,E_{\bar{i}_1^k}=0)\notag\\
&=\left(\sum_{k=m}^{n}\ \sum_{\{i_1,\dots, i_k\}\in B_{1}(n,k)}+\sum_{k=0}^{b(n,m)}\sum_{\{i_1,\dots, i_k\}\in B_{2}(n,k)}\right)H(X_0|X_{i_1^k})P(E_0=1,E_{i_1^k}=1,E_{\bar{i}_1^k}=0)\notag\\
&=\sum_{k=0}^{n-m+1}\sum_{t=0}^{b(k-1,m)}\sum_{\{i_1^t\}\in B(k-1,t)}H(X_0|X_{A(k,i_1^t)})P(E_{A(k,i_1^t)}=1,E_{\bar{A}(k,i_1^t)}=0)+T(n),\label{part}
\end{align}
where
\begin{eqnarray*}
T(n)&=&\sum_{k=0}^{b(n,m)}\sum_{\{i_1^k\}\in B_{2}(n,k)}H(X_0|X_{i_1^k})P(E_0=1,E_{i_1^k}=1,E_{\bar{i}_1^k}=0).
\end{eqnarray*}
It follows from $H(X_0|X_{i_1^k})\le \log K$ that
$$
T(n) \le \sum_{k=0}^{b(n,m)}\sum_{\{i_1^k\}\in B_{2}(n,k)}P(E_0=1,E_{i_1^k}=1,E_{\bar{i}_1^k}=0) \log K \le P(F_n) \log K,
$$
where $F_n$ is the event that ``there is no $m$ consecutive $1$'s in $E_{-n}^{-1}$''. Now, let $W_{i}=(E_{i},\cdots,E_{i-m+1})$ for $i\le -1$. Then it follows from the assumption that $W_{i}$ is also a stationary and ergodic process with $P(W_i=(1,\cdots,1))>0$. Using Poincare's recurrence theorem~\cite{durrettprobability}, we have that $P(W_{i}=(1,\cdots,1) \,\, i.o.)=1$, which implies that
$P(F)=0$, where $F$ denotes the event that ``there is no $m$ consecutive $1$'s in $E_{-\infty}^{-1}$''. This, together with the fact that $\lim_{n \to \infty} P(F_n)=P(F)$, implies that $\lim_{n \to \infty} T(n)=0$, and therefore the proof of the theorem is complete.
\end{proof}

The following corollary can be readily deduced from Theorem~\ref{entropyformula}.
\begin{corollary}\label{memorylessec}
Assume that $\{E_n\}$ is i.i.d. and $\{X_{n}\}$ is an $m$-th order Markov chain. Then
\begin{eqnarray}\label{high}
I(X;Y)&=&(1-\eps)^{m+1}\sum_{k=0}^{\infty}\sum_{t=0}^{b(k-1,m)}a(k,t)(1-\eps)^{t}\eps^{k-t},
\end{eqnarray}
where
\begin{equation*}
a(k,t)=\sum_{\{i_{1}\dots i_{t}\}\in B_{2}(k-1,t)}H(X_{0}|X_{i_{1}^{t}},X_{-k-m}^{-k-1}).
\end{equation*}
In particular, if $\{X_{n}\}$ is a first-order Markov chain,
\begin{equation}\label{first}
I(X;Y)=(1-\eps)^{2}\sum_{k=0}^{\infty}H(X_{0}|X_{-k-1})\eps^{k}.
\end{equation}
\end{corollary}
\begin{rem}\label{verduerasure}
A series expansion type formula for $H(X|Y)$ different from (\ref{first}) is given in Theorem $12$ of~\cite{vwerasure} for a discrete memoryless erasure channel with a first-order input Markov chain. It can be verified that these two formulas are ``equivalent'' in the sense that either one can be deduced from the other one via simple derivations. The form that our formula takes however makes it particularly effective for the capacity analysis of an input-constrained erasure channel.
\end{rem}

\section{Input-Constrained Memoryless Erasure Channel} \label{general-asymptotics}

In this section, we will focus on the case when $\{E_n\}$ is i.i.d. and $\mathcal{S}$ is an irreducible finite-type constraint of topological order $m$. With Lemma~\ref{nformula} and Corollary~\ref{memorylessec} established, we are ready to characterize the asymptotics of the capacity of this type of erasure channels.

As mentioned in Section~\ref{introduction-section}, when $\eps=0$, it is well known~\cite{parryintrinsicmarkovchains} that there exists an $m$th-order Markov chain $\hat{X}$ with $\mathcal{A}(\hat{X})=\mathcal{S}$ such that
\begin{equation} \label{Parry}
H(\hat{X})=H(\hat{X}_0|\hat{X}_{-m}^{-1})=\max_{\mathcal{A}(X) \subseteq \mathcal{S}}H(X)=C(\mathcal{S}, 0),
\end{equation}
where the maximization is over all stationary processes supported on $\mathcal{S}$. The following theorem characterizes the asymptotics of $C(\mathcal{S}, \eps)$ near $\eps=0$.
\begin{theorem}\label{asyerasurec}
Assume that $\{E_n\}$ is i.i.d. Then,
\begin{equation}\label{highsnrasy}
C(\mathcal{S},\eps)=C(\mathcal{S},0)-\left\{(m+1)H(\hat{X}_0|\hat{X}_{-m}^{-1})-\sum_{i=1}^{m}H(\hat{X}_{0}|\hat{X}_{-i+1}^{-1},\hat{X}_{-i-m}^{-i-1})\right\}\eps+O(\eps^2).
\end{equation}
Moreover, for any $n \geq m$, $C^{(n)}(\mathcal{S}, \eps)$ is of the same asymptotic form as in (\ref{asyerasurec}), namely,
\begin{equation}\label{highsnrasy-1}
C^{(n)}(\mathcal{S},\eps)=C(\mathcal{S}, 0)-\left\{(m+1)H(\hat{X}_0|\hat{X}_{-m}^{-1})-\sum_{i=1}^{m}H(\hat{X}_{0}|\hat{X}_{-i+1}^{-1},\hat{X}_{-i-m}^{-i-1})\right\}\eps+O(\eps^2).
\end{equation}
\end{theorem}

\begin{proof}
To establish~(\ref{highsnrasy}), we prove that $C(\mathcal{S}, \eps)$ is lower and upper bounded by the same asymptotic form as in (\ref{highsnrasy}).

For the lower bound part, we consider the channel (\ref{mec}) with $\hat{X}$ as its input. Note that
$$
P(\hat{F}_0)=(1-\eps)^{m} \mbox{ and } P(\hat{F}_k)=\eps(1-\eps)^{m}\quad \mbox{for $1\le k\le m$},
$$
and furthermore, for $k\ge m+1$
$$
P(\hat{F}_k)=\sum_{t=0}^{b(k-1,m)} |B_{2}(k-1,t)| (1-\eps)^{t+m}\eps^{k-t},
$$
where we have defined
$$
\hat{F}_{k}=\{E_{-k-1}^{-k-m}=1, E_{-k}=0, E_{-k+1}^{-1}\ \mbox{contains no $m$ consecutive $1$'s}\}.
$$
It then follows that
\begin{eqnarray}
&& \hspace{-2cm} (1-\eps)^{m+1}\sum_{k=m+1}^{\infty}\sum_{t=0}^{b(k-1,m)}a(k,t)(1-\eps)^{t}\eps^{k-t}\notag\\
%&=&(1-\eps)^{m+1}\sum_{k=m+1}^{\infty}\sum_{t=0}^{b(k-1,m)}a(k,t)(1-\eps)^{t}\eps^{k-t}\notag\\
&\leq &(1-\eps) \sum_{k=m+1}^{\infty}P(\hat{F}_{k}) \log K\notag\\
&\stackrel{(a)}{=}&(1-\eps) (1-\sum_{k=0}^{m}P(\hat{F}_{k}) \log K\notag\\
&=&(m\eps^2(1-\eps)^m+\sum_{u=2}^{m}{m\choose u}(1-\eps)^{m-u}\eps^u)\log K \notag\\
&=&O(\eps^2), \label{tailestimate}
\end{eqnarray}
where $(a)$ follows from $P(\cup_{k\ge 0} \hat{F}_k)=1$ and the constant in $O(\eps^2)$ depends only on $m$ and $K$.
Then, from Corollary~\ref{memorylessec}, it follows that
\begin{align}\label{lowerbound}
C(\mathcal{S}, \eps)&\ge I(\hat{X};Y)=(1-\eps)^{m+1}\sum_{k=0}^{\infty}\sum_{t=0}^{b(k-1,m)}a(k,t)(1-\eps)^{t}\eps^{k-t}\notag\\
&=(1-\eps)^{m+1}\sum_{k=0}^{m}\sum_{t=0}^{b(k-1,m)}a(k,t)(1-\eps)^{t}\eps^{k-t}+(1-\eps)^{m+1}\sum_{k=m+1}^{\infty}\sum_{t=0}^{b(k-1,m)}a(k,t)(1-\eps)^{t}\eps^{k-t}\notag\\
&\stackrel{(b)}{=}H(\hat{X}_0|\hat{X}_{-m}^{-1})+\left\{(m+1)H(\hat{X}_0|\hat{X}_{-m}^{-1})-\sum_{i=1}^{m}H(\hat{X}_{0}|\hat{X}_{-i+1}^{-1},\hat{X}_{-i-m}^{-i-1})\right\}\eps+O(\eps^2), \notag
\end{align}
where $(b)$ follows from (\ref{tailestimate}) and
{\small
$$
\hspace{-1cm} \sum_{k=0}^{m}\sum_{t=0}^{b(k-1,m)}a(k,t)(1-\eps)^{t+m+1}\eps^{k-t}=H(\hat{X}_0|\hat{X}_{-m}^{-1})+\left\{(m+1)H(\hat{X}_0|\hat{X}_{-m}^{-1})-\sum_{i=1}^{m}H(\hat{X}_{0}|\hat{X}_{-i+1}^{-1},\hat{X}_{-i-m}^{-i-1})\right\}\eps+O(\eps^2).
$$}
\noindent This, together with (\ref{Parry}), establishes that $C(\mathcal{S}, \eps)$ is lower bounded by the asymptotic form in (\ref{highsnrasy}).

For the upper bound part, we will adapt the argument in~\cite{asymptotics-binary}. Let
$$
S_n=\left\{ \mathbf{p}_n=(p(\hat{x}_{-n}^0):\hat{x}_{-n}^0\in \mathcal{A}(\hat{X}_{-n}^0)):p(\hat{x}_{-n}^0)>0,\sum_{\hat{x}_0^n\in \mathcal{A}(\hat{X}_{-n}^0)}p(\hat{x}_{-n}^0)=1 \right\}
$$
and
$$
S_{n,\delta}=\{\mathbf{p}_n\in S_n:p(\hat{x}_{-n}^{0})>\delta \mbox{ for any } \hat{x}_{-n}^{0}\in \mathcal{A}(\hat{X}_{-n}^0)\},
$$
where $\mathcal{A}(\hat{X}_{-n}^0)=\{\hat{x}_{-n}^0: p(\hat{x}_n^0)>0\}$.
In this proof, we define
$$
C_{n}(\eps,\mathcal{S})=\sup_{\p_n\in S_n}H(Y_0|Y_{-n}^{-1})-H(\eps).
$$
It then follows from Lemma~\ref{nformula} that
$$
C_{n}(\mathcal{S}, \eps) =\sup_{\p_n\in S_n} f(\p_n,\eps),
$$
where
$$
f(\p_n,\eps) \triangleq \sum_{k=1}^{n}\sum_{D \subseteq [-n, -1],|D|=k} H(X_0|X_{D})(1-\eps)^{k+1}\eps^{n-k}.
$$
Let $\overline{\p}_n(\eps)$ maximize $f(\p_n,\eps)$. As $f(\p_n,\eps)$ is continuous in $(\p_n,\eps)$ and is maximized at $\hat{\p}_n$ when $\eps=0$, there exists some $\eps_0>0$ (depends on $n$) and $\delta>0$ such that for all $\eps<\eps_0$, $\overline{\p}_n(\eps)\in S_{n,\delta}$. Then for $\eps\le \eps_0$, there exists some constant $M$ (depends on $n$) such that
$$
C_{n}(\mathcal{S}, \eps)\le \sup_{\p_n\in S_{n,\delta} } \left\{H(X_0|X_{-n}^{-1})+\left((n+1)H(X_0|X_{-n}^{-1})-\sum_{k=1}^{n}H(X_0|X_{-k+1}^{-1},X_{-n}^{-k-1})\right)\eps\right\}+M \eps^2.
$$
From now on, we write
$$
g_1(\p_n)=H(X_0|X_{-n}^{-1}), \quad g_2(\p_n)=\left((n+1)H(X_0|X_{-n}^{-1})-\sum_{k=1}^{n}H(X_0|X_{-k+1}^{-1},X_{-n}^{-k-1})\right)\eps.
$$
Letting $\mathbf{H}=\mathbf{H}(\p_n)$ be the Hessian of $g_1(\p_n)$, we now expand $g_1(\p_n)$ and $g_2(\p_n)$ around $\mathbf{p}_n=\hat{\p}_n$ to obtain
$$
g_1(\p_n)=g_1(\hat{\p}_n)+\frac{1}{2}{\mathbf q}_n^{T}\mathbf{H}{\mathbf q}_n+O(|{\mathbf q}_n|^2)
$$
and
$$
g_2(\p_n)=g_2(\hat{\p}_n)+\sum_{\hat{x}_{-n}^0\in \mathcal{A}(\hat{X}_{-n}^0)}\frac{\partial g_2(\hat{\p}_n)}{\partial p(\hat{x}_{-n}^0)}q(\hat{x}_{-n}^0)+O(|{\mathbf q}_n|).
$$
where ${\mathbf q}_n \triangleq \p_n-\hat{\p}_n$ contains all $q(\hat{x}_{-n}^0)$ as its coordinates. Since $\mathbf{H}$ is negative definite (see Lemma 3.1~\cite{asymptotics-binary}), we deduce that, for $|{\mathbf q}_n|$ sufficiently small,
$$
g_1(\p_n)+g_2(\p_n)\eps\le g_1(\hat{\p}_n)+g_2(\hat{\p}_n)\eps+\frac{1}{4}{\mathbf q}_n^{T}\mathbf{H}{\mathbf q}_n+2\sum_{\hat{x}_{-n}^0\in \mathcal{A}(\hat{X}_{-n}^0)}\left| \frac{\partial g_2(\hat{\p}_n)}{\partial p(\hat{x}_{-n}^0)}q(\hat{x}_{-n}^0)\right|\eps.
$$
Without loss of generality, we henceforth assume $\mathbf{H}$ is a diagonal matrix with all diagonal entries denoted $k(\hat{x}_{-n}^0) < 0$ (since otherwise we can diagonalize $\mathbf{H}$). Now, let
$$
\mathcal{A}_1(\hat{X}_{-n}^0)=\left\{\hat{x}_{-n}^0: \frac{1}{4}k(\hat{x}_{-n}^0)\left|q(\hat{x}_{-n}^0)\right|^2+2\left|\frac{\partial g_2(\hat{\p}_n)}{\partial p(\hat{x}_{-n}^0)}\right|\left|q(\hat{x}_{-n}^0)\right|\eps<0 \right\}
$$
and let $\mathcal{A}_2(\hat{X}_{-n}^0)$ denote the complement of $\mathcal{A}_1(\hat{X}_{-n}^0)$ in $\mathcal{A}(\hat{X}_{-n}^0)$:
$$
\mathcal{A}_2(\hat{X}_{-n}^0)=\mathcal{A}(\hat{X}_{-n}^0)-\mathcal{A}_1(\hat{X}_{-n}^0).
$$
Then, we have
\begin{align}
f(\p_n,\eps)&\le  g_1(\hat{\p}_n)+g_2(\hat{\p}_n)\eps+\frac{1}{4}{\mathbf q}_n^{T} \mathbf{H} {\mathbf q}_n+2\sum_{\hat{x}_{-n}^0\in \mathcal{A}(\hat{X}_{-n}^0)}\left|\frac{\partial g_2(\hat{\p}_n)}{\partial p(\hat{x}_{-n}^0))}\right|\left|q(\hat{x}_{-n}^0)\right|\eps+M\eps^2\notag\\
&\le g_1(\hat{\p}_n)+g_2(\hat{\p}_n)\eps+\sum_{\hat{x}_{-n}^0\in \mathcal{A}_2(\hat{X}_{-n}^0)}\left(\frac{1}{4}q(\hat{x}_{-n}^0)^2k(\hat{x}_{-n}^0)+2\left|\frac{\partial g_2(\hat{\p}_n)}{\partial p(\hat{x}_{-n}^0)}\right|\left|q(\hat{x}_{-n}^0)\right|\eps\right)+M\eps^2\notag\\
&\stackrel{(a)}{\le} g_1(\hat{\p}_n)+g_2(\hat{\p}_n)\eps+\sum_{\hat{x}_{-n}^0\in \mathcal{A}_2(\hat{X}_{-n}^0)}\frac{4\left|\frac{\partial g_2(\hat{\p}_n)}{\partial p(\hat{x}_{-n}^0)}\right|^2\eps^2}{-k(\hat{x}_{-n}^0)}+M\eps^2\notag,
\end{align}
where $(a)$ follows from the easily verifiable fact that for any $\hat{x}_{-n}^0\in \mathcal{A}_2(\hat{X}_{-n}^0)$,
$$
\frac{1}{4}q(\hat{x}_{-n}^0)^2k(\hat{x}_{-n}^0)+2\left|\frac{\partial g_2(\hat{\p}_n)}{\partial p(\hat{x}_{-n}^0)}\right|\left|q(\hat{x}_{-n}^0)\right|\eps\le \frac{4\left|\frac{\partial g_2(\hat{\p}_n)}{\partial p(\hat{x}_{-n}^0)}\right|^2\eps^2}{-k(\hat{x}_{-n}^0)}.
$$
Since $\hat{X}_{-n}^0$ is an $m$-th order Markov chain, we deduce that
\begin{align*}
g_1(\hat{\p}_n)+g_2(\hat{\p}_n)\eps&=H(\hat{X}_0|\hat{X}_{-n}^{-1})+\left((n+1)H(\hat{X}_0|\hat{X}_{-n}^{-1})-\sum_{k=1}^{n}H(\hat{X}_0|\hat{X}_{-k+1}^{-1},\hat{X}_{-n}^{-k-1})\right)\eps\\
&=H(\hat{X}_0|\hat{X}_{-m}^{-1})+\left((m+1)H(\hat{X}_0|\hat{X}_{-m}^{-1})-\sum_{k=1}^{m}H(\hat{X}_0|\hat{X}_{-k+1}^{-1},\hat{X}_{-m}^{-k-1})\right)\eps,
\end{align*}
We now ready to deduce that, for some positive constant $M_1$,
$$
f(\mathbf{p}_n,\eps)\le H(\hat{X}_0|\hat{X}_{-m}^{-1})+\left((m+1)H(\hat{X}_0|\hat{X}_{-m}^{-1})-\sum_{k=1}^{m}H(\hat{X}_0|\hat{X}_{-k+1}^{-1},\hat{X}_{-m}^{-k-1})\right)\eps+M_1\eps^2,
$$
which further implies that
$$
C(\eps,\mathcal{S})\le C_{n}(\eps,\mathcal{S})\le H(\hat{X}_0|\hat{X}_{-m}^{-1})+\left((m+1)H(\hat{X}_0|\hat{X}_{-m}^{-1})-\sum_{k=1}^{m}H(\hat{X}_0|\hat{X}_{-k+1}^{-1},\hat{X}_{-m}^{-k-1})\right)\eps+M_1\eps^2.
$$
The proof of (\ref{highsnrasy}) is then complete.

With $C(\mathcal{S}, \eps)$ replaced with $C^{(m)}(\mathcal{S}, \eps)$, the proof of (\ref{highsnrasy}) also establishes (\ref{highsnrasy-1}).
\end{proof}

\begin{rem}
In a fairly general setting (where input constraints are not considered), a similar asymptotic formula with a constant term, a term linear in $\eps$ and a residual $o(\eps)$-term has been derived in Theorem $23$ of~\cite{vwerasure}.
\end{rem}

As an immediate corollary of Theorem~\ref{asyerasurec}, the following result gives asymptotics of the capacity of a BEC($\eps$) with the input supported on the $(1, \infty)$-RLL constraint $\mathcal{S}_0$.
\begin{corollary}\label{asmpt}
Assume $K=2$ and $\{E_n\}$ is i.i.d. Then, we have
$$
C(\mathcal{S}_0, \eps)=\log\lambda-\frac{2\log 2}{1+\lambda^2}\eps+O(\eps^2),
$$
and for any $n \geq 1$, $C^{(n)}(\mathcal{S}_0, \eps)$ is of the same asymptotic form, namely,
\begin{equation} \label{first-order-partial}
C^{(n)}(\mathcal{S}_0, \eps)=\log\lambda-\frac{2\log 2}{1+\lambda^2}\eps+O(\eps^2).
\end{equation}
\end{corollary}

\begin{proof}
Let $\lambda=(1+\sqrt{5})/2$. It is well known~\cite{symbolicmarcuslind} that the noiseless capacity
$$
C(\mathcal{S}_0, 0)=\log \lambda
$$
and the first-order Markov chain $\{\hat{X}_n\}$ with the following transition probability matrix
\begin{equation}\label{achievingmatrix}
\Pi=\begin{bmatrix}
1-1/\lambda^2 &1/\lambda^2\\
1&0\\
\end{bmatrix},
\end{equation}
achieves the noiseless capacity, that is, $H(\hat{X})=C(\mathcal{S}_0, 0)=\log \lambda$. Furthermore, via straightforward computations, we deduce that
$$
H(\hat{X}_{0}|\hat{X}_{-2})=\frac{\lambda^{2}}{1+\lambda^2}H(2/\lambda^2)+\frac{1}{1+\lambda^2}H(1/\lambda^2)=2\log \lambda-\frac{2\log 2}{1+\lambda^2},
$$
which, together with the fact that
$$
H(\hat{X})=H(\hat{X}_{0}|\hat{X}_{-1})=\log\lambda,
$$
implies
$$
C(\mathcal{S}_0, \eps)=\log\lambda-\frac{2\log 2}{1+\lambda^2}\eps+O(\eps^2),
$$
as desired.

And the asymptotic form of $C^{(n)}(\mathcal{S}_0, \eps)$ follows from similar computations.
\end{proof}

\begin{rem}
The asymptotic form in (\ref{first-order-partial}) only gives partial asymptotics of $C^{(n)}(\mathcal{S}_0, \eps)$ for any $n \geq 1$. For the case $n=1$, we will derive later on the full asymptotics of $C^{(1)}(\mathcal{S}_0, \eps)$; see (\ref{first-order-full}) in Section~\ref{binary-asymptotics}.

\end{rem}

\section{Input-Constrained Binary Erasure Channel}  \label{binary-asymptotics}

In this section, we will focus on a BEC($\eps$) with the input being a first-order Markov process supported on the $(1, \infty)$-RLL constraint $\mathcal{S}_0$. To be more precise, we assume that $K=2$, $\{E_n\}$ is i.i.d. and $\{X_{n}\}$ is a first-order Markov chain, taking values in $\{1,2\}$ and having the following transition probability matrix:
$$
\Pi=\begin{bmatrix}
1-\theta&\theta\\
1&0\\
\end{bmatrix}.
$$
In Section~\ref{sub-1}, we will show that $I(X; Y)$ is concave with respect to $\theta$, and in Section~\ref{sub-2}, we apply the algorithm in~\cite{randomapproachhan} to numerically evaluate $C^{(1)}(\mathcal{S}_0, \eps)$, whose convergence is guaranteed by the above-mentioned concavity result. Finally, in Section~\ref{sub-3}, we characterize the full asymptotics of $C^{(1)}(\mathcal{S}_0, \eps)$ around $\eps=0$.

\subsection{Concavity}  \label{sub-1}

The concavity of the mutual information rate of special families of finite-state machine channels (FSMCs) has been considered in~\cite{hm09b} and~\cite{lihan2013}. The results therein actually imply that the concavity of $I(X; Y)$ with respect to some parameterization of the Markov chain $X$ when $\eps$ is small enough. In this section, however, we will show that $I(X; Y)$ is concave with respect to $\theta$, irrespective of the values of $\eps$. Below is the main theorem of this section.
\begin{theorem}\label{concavitybec}
For all $\eps\in [0,1)$, $I(X;Y)$ is strictly concave with respect to $\theta$, $0 \leq \theta \leq 1$.
\end{theorem}

\begin{proof}
From Corollary~\ref{memorylessec}, it follows that to prove the theorem, it suffices to show that for any $n \geq 1$, $H(X_{0}|X_{-n})$ is strictly concave with respect to $\theta$, $0 \leq \theta \leq 1$. To prove this, we will deal with the following several cases:

{\bf \underline{Case 1: $n=1$.}} Straightforward computations give
$$
H(X_{0}|X_{-1})=\frac{-\theta\log \theta-(1-\theta)\log (1-\theta)}{1+\theta}
$$
and
$$
\hspace{-3mm} H^{\prime\prime}(X_{0}|X_{-1})=\frac{1}{(1+\theta)^3}\left\{2\log \theta-4\log(1-\theta)-\frac{1}{\theta}-\frac{4}{1-\theta}+1\right\}.
$$
\hspace{-1.9mm} One checks that the function within the brace is negative and it takes the maximum at $\theta={1}/{2}$. Therefore $H(X_{0}|X_{-1})$ is strictly concave in $\theta$.

{\bf \underline{{Case 2}: $n\ge 2$.}} By definition,
\begin{eqnarray*}H(X_{0}|X_{-n})&=&P(X_{-n}=1)H(X_{0}|X_{-n}=1)+P(X_{-n}=2)H(X_{0}|X_{-n}=2).
\end{eqnarray*}
The following facts can be verified easily:
\begin{itemize}
\item[(i)] $f(\theta) \triangleq P(X_{-n}=1)=1-P(X_{-n}=2)=\frac{1}{1+\theta}$;
%\item[(ii)] $f^{\prime}(p)=-\frac{1}{(1+p)^{2}},$ $f^{\prime\prime}(p)=\frac{2}{(1+p^{3})}$;
\item[(ii)]
the $n$-step transition probability matrix of the Markov chain $\{X_n\}$ is
 \begin{align*}
\Pi^{n}%&=\begin{bmatrix}
%\frac{1-(-\theta)^{n+1}}{1+\theta}&\frac{\theta+(-\theta)^{n+1}}{1+\theta}\\
%\frac{1-(-\theta)^{n}}{1+\theta}&\frac{\theta+(-\theta)^{n}}{1+\theta}\\
%\end{bmatrix}\\
&=\begin{bmatrix}
g_{n+1}(\theta)&1-g_{n+1}(\theta)\\
g_{n}(\theta)&1-g_{n}(\theta)
\end{bmatrix},
\end{align*}
where $$g_{n}(\theta) \triangleq \frac{1-(-\theta)^{n}}{1+\theta};$$

\item[(iii)] $H(y)=-y\log y-(1-y)\log(1-y)$ is strictly concave with respect to $y$ for $y\in(0,1)$.
\end{itemize}
With the above notation, we have
$$
H(X_{0}|X_{-n})=f(\theta)H(g_{n+1}(\theta))+(1-f(\theta))H(g_{n}(\theta))
$$
and
\begin{align}
H^{\prime\prime}(X_{0}|X_{-n})&=fH^{\prime\prime}(g_{n+1})(g_{n+1}^{\prime})^2+(1-f)H^{\prime\prime}(g_{n})(g_{n}^{\prime})^2\label{term1}\\
&{}\hspace{4.5mm}+f^{\prime\prime} (H(g_{n+1})-H(g_{n}))\label{term2}\\
&{}\hspace{4.5mm}-2f^{\prime}H^{\prime}(g_{n})g_{n}^{\prime}+(1-f)H^{\prime}(g_{n})g_{n}^{\prime\prime}\label{term4}\\
&{}\hspace{4.5mm}+2f^{\prime}H^{\prime}(g_{n+1})g_{n+1}^{\prime}+fH^{\prime}(g_{n+1})g_{n+1}^{\prime\prime}\label{term3},
\end{align}
where $f=f(\theta)$ and $g_{n}=g_{n}(\theta)$.
It follows from (i) and (iii) that the term (\ref{term1}) is strictly negative. So, to prove the theorem, it suffices to show that $T \triangleq (\ref{term2})+ (\ref{term4})+(\ref{term3})\le 0$.

By the mean value theorem,
\vspace{-2mm}
\begin{eqnarray*}
(\ref{term2})&=&\frac{2}{(1+\theta)^{3}}(g_{n+1}-g_{n})\log\frac{1-z_1}{z_1}\\
&=&\frac{2(-\theta)^{n}(1+\theta)}{(1+\theta)^{4}}\log\frac{1-z_1}{z_1},
\end{eqnarray*}
where $z_1$ lies between $g_{n}$ and $g_{n+1}$. As a function of $z_1$, $\frac{2(-\theta)^{n}(1+\theta)}{(1+\theta)^{4}}\log\frac{1-z_1}{z_1}$ takes the maximum at $g_{n}$. It then follows that
\begin{eqnarray*}
\hspace{-1mm} T&\le&\frac{c_{n}\{2\theta-2+(-\theta)^{n-1}[(n^2-3n)\theta^2+2(n^2-n-2)+n^2+n]\}}{(1+\theta)^{4}}\\
&{}&+\frac{c_{n+1}\{4-(-\theta)^{n-1}[(n^2-3n)\theta^2+2(n^2-n-2)+n^2+n]\}}{(1+\theta)^{4}}\\
&=&\frac{c_{n}[2\theta-2+Q(n,\theta)(-\theta)^{n-1}]}{(1+\theta)^{4}}+\frac{c_{n+1}[4-Q(n,\theta)(-\theta)^{n-1}]}{(1+\theta)^{4}},
\end{eqnarray*}
\noindent where $$Q(n,\theta)=(n^2-3n)\theta^2+2(n^2-n-2)\theta+n^2+n$$ and $$c_{n}=\log\frac{1-g_n}{g_n}.$$
We then consider the following several cases:
\begin{item}

{\bf \underline{Case 2.1: $n$ is a positive even number.}}
We first consider the case that $g_{n}\le\frac{1}{2}$. For this case, obviously we have $c_{n}\ge 0$, $g_{n+1}>\frac{1}{2}$ and $Q(n,\theta)>0$, which further implies that
\begin{align*}
{T}\le&\frac{c_{n}[2\theta-2-Q(n,\theta)\theta^{n-1}]}{(1+\theta)^{4}}+\frac{c_{n+1}[4+Q(n,\theta)\theta^{n-1}]}{(1+\theta)^{4}}<0.
\end{align*}
Now, for the case that $g_{n}> \frac{1}{2}$, again obviously we have $c_{n}<0$ and furthermore,
$$
c_{n+1}-c_{n}=\log\frac{1-g_{n+1}}{g_{n+1}}-\log\frac{1-g_{n}}{g_{n}}\le 0,
$$
where we have used the fact that $g_n \leq g_{n+1}$ for a positive even number $n$. Now, we are ready to deduce that
\begin{align*}
T&\le\frac{c_{n}[2\theta-2-Q(n,\theta)\theta^{n-1}]}{(1+\theta)^{4}}+\frac{c_{n+1}[4+Q(n,\theta)\theta^{n-1}]}{(1+\theta)^{4}}\\
&=\frac{(c_{n+1}-c_{n})(4+Q(n,\theta)\theta^{n-1})}{(1+\theta)^{4}}+\frac{c_{n}(2\theta+2)}{(1+\theta)^{4}}\\
&< 0.
\end{align*}
{\bf \underline{Case 2.2: $n$ is a positive odd integer and $n\ge 3$.}}
In this case, we have
\begin{align*}
T&\le\frac{c_{n}[2\theta-2+Q(n,\theta)\theta^{n-1}]}{(1+\theta)^{4}}+\frac{c_{n+1}[4-Q(n,\theta)\theta^{n-1}]}{(1+\theta)^{4}}\\
&=\frac{(c_{n+1}-c_{n})[4-Q(n,\theta)\theta^{n-1}]}{(1+\theta)^{4}}+\frac{c_{n}(2\theta+2)}{(1+\theta)^{4}}\\
&\le\frac{1}{(1+\theta)^{4}}\left\{\frac{[4-Q(n,\theta)\theta^{n-1}]\theta^{n}}{(1-y_{2})y_{2}}+(2\theta+2)(1/g_{n}-2)\right\},
\end{align*}
where the last inequality follows from the mean value theorem, the inequality $\log z \le z-1$ for $z>0$ and the fact that $z_{2}$ lies between $g_{n}$ and $g_{n+1}$.
\end{item}

Let $y_{2}=B_n/(1+\theta)$ and \begin{small}
$C_{n}=1+\theta-B_{n}$\end{small}, where
$B_{n}\in [1-\theta^{n+1},1+\theta^{n}]$. Then
\begin{eqnarray*}
T&\le&\frac{1}{(1+\theta)^{4}}\left\{\frac{[4-Q(n,\theta)\theta^{n-1}]\theta^{n}}{(1-g_{2})g_{2}}+(2\theta+2)(1/g_{n}-2)\right\}\\
&\le&\frac{(2\theta-2-4\theta^n)B_{n}C_{n}+(1+\theta)(1+\theta^n)\theta^n(4-Q(n,\theta)\theta^{n-1})}{B_{n}C_{n}(1+\theta)^{3}(1+\theta^n)}.
\vspace{-3mm}
\end{eqnarray*}

Note that the above numerator, as a function of $B_n$,
takes the maximum at $B_n=1+\theta^n$. Denote this maximum by $-2\theta(1+\theta^n)h_{1}(n,\theta)$, where
\begin{eqnarray*}
h_{1}(n,\theta)&=&1-\theta-3\theta^{n-1}+\theta^n+\frac{Q(n,\theta)}{2}\theta^{2n-2}
+\frac{Q(n,\theta)-4}{2}\theta^{2n-1}.
\end{eqnarray*}

To complete the proof, it suffices to prove $h_1(n,\theta)\ge 0$. Substituting $Q(n,\theta)$ into $h_{1}(n,\theta)$, we have
\begin{eqnarray*}
h_{1}(n,\theta)&=& 1-\theta-3\theta^{n-1}+\theta^n+\frac{Q(n,\theta)}{2}\theta^{2n-2}
+\frac{Q(n,\theta)-4}{2}\theta^{2n-1}\\
&=&1-\theta-3\theta^{n-1}+\theta^n+\frac{\theta^{2n-2}}{2}[(n^2-3n)\theta^3\\
&{}&+ (3n^2-5n-4)\theta^2+(3n^2-n-8)\theta+n^2+n]\\
&\ge&h(n,\theta),
\end{eqnarray*}
where
$$h(n,\theta)=1-\theta-3\theta^{n-1}+\theta^n+(4n^2-4n-6)\theta^{2n+1}.$$

The following facts can be easily verified:
\begin{itemize}
\item[(a)] $h(n,\theta)$ takes the minimum at some $\theta_0$, where $\theta_0$ satisfies the following equation
{\begin{equation}\hspace{-6mm}
h^{\prime}(n,\theta)=0;\label{theta0}
\end{equation}}
\item[(b)] $h^{\prime}(n,\theta)<0$ for $\theta\in [0,1/2]$ and $n\ge 11$.
\end{itemize}
It then follows from (a) and (b) that $\theta_{0}\ge 1/2$ for $n>11$.
Solving~(\ref{theta0}) in $\theta^{n-2}$, we have
$$
\hspace{0.081cm} \theta_{0}^{n-2}=\frac{(3-\theta_0)n-3+\sqrt{((3-\theta_0)n-3)^2+4(2n+1)(4n^2-4n-6)\theta_{0}^{5}}}{2(2n+1)(4n^2-4n-6)\theta_{0}^{5}}.
$$

For $n\ge 11$, substituting $\theta_{0}^{n-2}$ into $h(n,\theta)$, we have
\begin{align*}
h_{1}(n,\theta)&\ge h(n,\theta)\\
&\ge h(n,\theta_0)\\
&=1-\theta_{0}-3\theta^{n-1}_{0}+\theta_{0}^{n}+(4n^2-4n-6)\theta_{0}^{2n+1}\\
&\ge \theta_{0}^{n-1}[-3+\theta_0+(4n^2-4n-6)\theta_{0}^{4}\cdot \theta_{0}^{n-2}]\\
&\ge\theta_{0}^{n-1} v(n),
\end{align*}
where
\begin{align*}v(n)=-\frac{5}{2}+\frac{2n-3+\sqrt{(2n-3)^2+(2n+1)(4n^2-4n-6)2^{-3}}}{2(2n+1)}.
\end{align*}
One checks that $v(n)>0$ for $n\ge 65$. Now, with the fact that $h_{1}(n,\theta)>0$ for $3\le n\le 65$ (this can be verified via tedious yet straightforward computations since $h_{1}(n,\theta)$ is an elementary function), we conclude that for $h_{1}(n,\theta)\ge 0$ for all $n\ge 3$ and $\theta\in [0,1]$.
\end{proof}

\subsection{Numerical evaluation of $C^{(1)}(\mathcal{S}_0, \eps)$} \label{sub-2}

When $\eps=0$, that is, when the channel is ``perfect'' with no erasures, both $C(\mathcal{S}, \eps)$ and $C^{(m)}(\mathcal{S}, \eps)$ boil down to the noiseless capacity of the constraint $\mathcal{S}$, which can be explicitly computed~\cite{parryintrinsicmarkovchains}; however, little progress has been made for the case when $\eps > 0$ due to the lack of simple and explicit characterization for $C(\mathcal{S},\eps)$ and $ C^{(m)}(\mathcal{S},\eps)$. In terms of numerically computing $C(\mathcal{S}, \eps)$ and $C^{(m)}(\mathcal{S}, \eps)$, relevant work can be found in the subject of FSMCs, as input-constrained memoryless erasure channels can be regarded as special cases of FSMCs. Unfortunately, the capacity of an FSMC is still largely unknown and the fact that our channel is only a special FSMC does not seem to make the problem easier.

Recently, Vontobel {\em et al.}~\cite{vontobel} proposed a generalized Blahut-Arimoto algoritm (GBAA) to compute the capacity of an FSMC; and in~\cite{randomapproachhan}, Han also proposed a randomized algorithm to compute the capacity of an FSMC. For both algorithms, the concavity of the mutual information rate is a desired property for the convergence (the convergence of the GBAA requires, in addition, the concavity of certain conditional entropy rate). On the other hand, as elaborated in~\cite{lihan2013}, such a desired property, albeit established for a few special cases~\cite{hm09b, lihan2013}, is not true in general.

The concavity established in the previous section allows us to numerically compute $C^{(1)}(\mathcal{S}_0, \eps)$ using the algorithm in~\cite{randomapproachhan}. The randomized algorithm proposed in~\cite{randomapproachhan} iteratively compute $\{\theta_{n}\}$ in the following way:
$$\theta_{n+1}=\begin{cases}
\theta_{n},&\mbox{if}\ \theta_{n}+a_{n}g_{n^b}(\theta_{n})\in [0,1],\\
\theta_{n}+a_{n}g_{n^b}(\theta_{n}), &\mbox{otherwise},
\end{cases}
$$
where $g_{n^{b}}(\theta_{n})$ is a simulator for $I^{\prime}(X;Y)$ (for details, see~\cite{randomapproachhan}). The author shows that $\{\theta_{n}\}$ converges to the first-order capacity-achieving distribution if $I(X;Y)$ is concave with respect to $\theta$, which has been proven in Theorem~\ref{concavitybec}. Therefore, with proven convergence, this algorithm can be used to compute the first-order capacity-achieving distribution $\theta(\eps)$ and the first-order capacity $C^{(1)}(\mathcal{S}_0, \eps)$ (in bits), which are shown in Fig.~\ref{capacity-achievingdistribution} and Fig.~\ref{capacity1}, respectively.
\begin{figure}[!t]
\begin{center}
\includegraphics[width=2.5in]{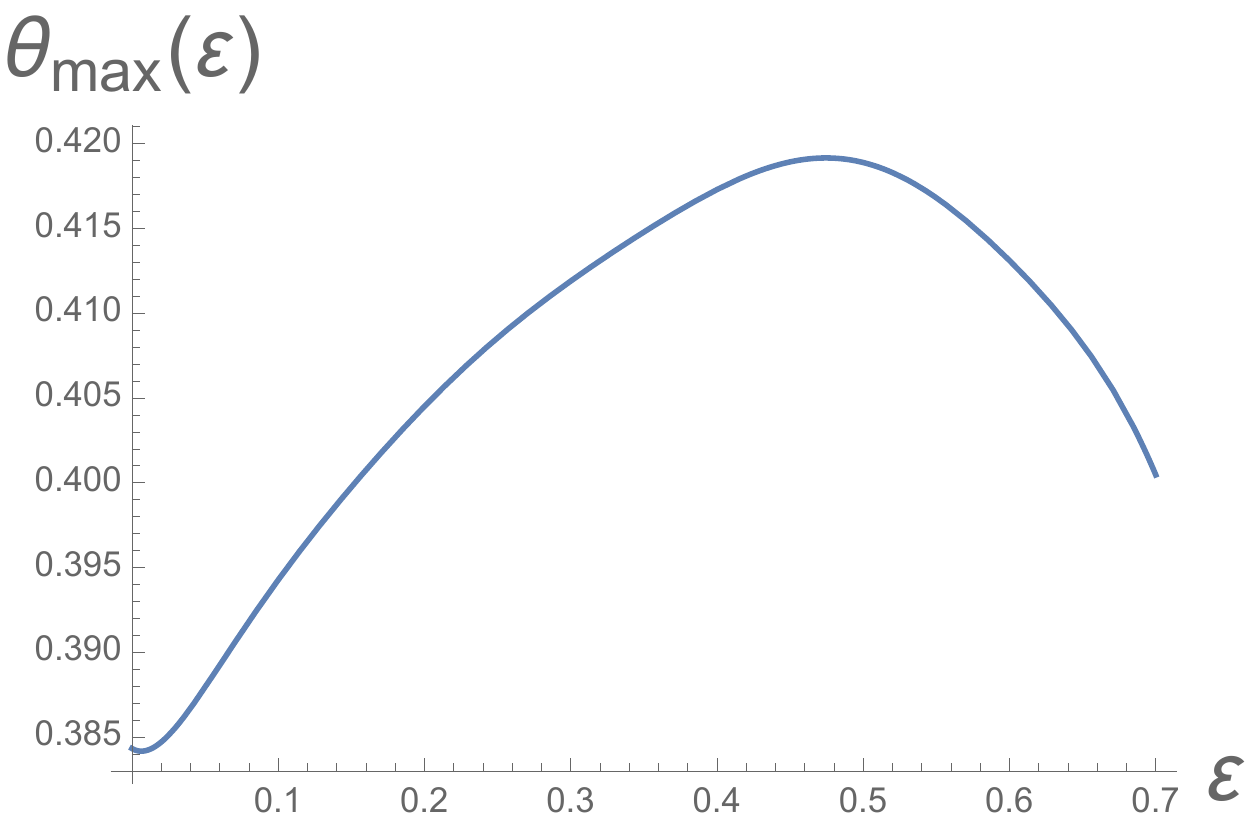}
\caption{Capacity-achieving Distribution}
\label{capacity-achievingdistribution}
\vspace{5mm}
\includegraphics[width=2.5in]{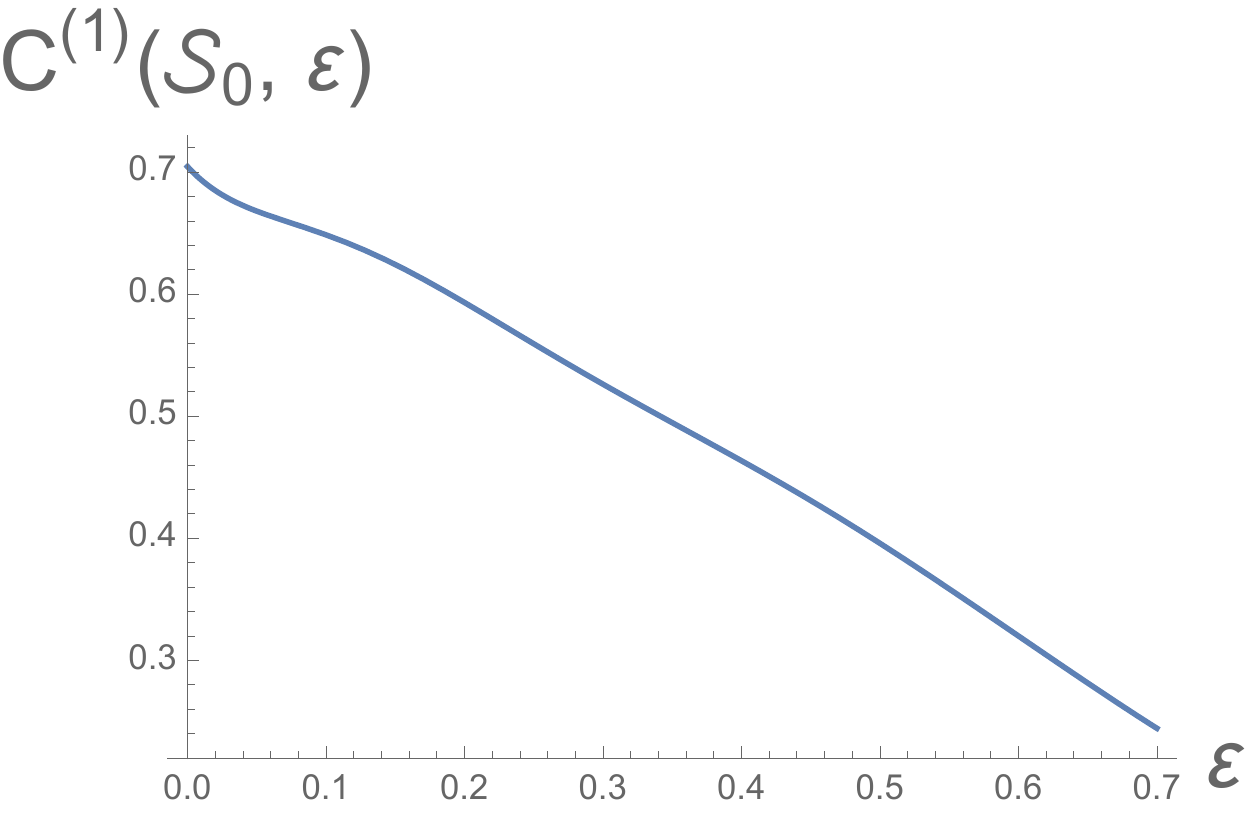}
\caption{Capacity}
\label{capacity1}
\end{center}
\end{figure}

\subsection{Full Asymptotics} \label{sub-3}

As in Section~\ref{sub-1}, the noiseless capacity of $(1,\infty)$-RLL constraint $\mathcal{S}_0$ is achieved by the first-order Markov chain with the transition probability matrix~(\ref{achievingmatrix}). So, we have
$$
\theta_{\mbox{\small max}}(\eps)=\argmax_{\theta} I(X;Y)=1/\lambda^2,
$$
where $\lambda=(1+\sqrt{5})/2$. In this section, we give a full asymptotic formula for $\theta_{max}(\eps)$ around $\eps=0$, which further leads to a full asymptotic formula for $C^{(1)}(\mathcal{S}_0, \eps)$ around $\eps=0$.

The following theorem gives the Taylor series of $\theta_{\mbox{\small max}}$ in $\eps$ around $\eps=0$, which leads to an explicit formula for the $n$-th derivative of $C^{(1)}(\mathcal{S}_0, \eps)$ at $\eps=0$, whose coefficients can be explicitly computed.

\begin{theorem}\label{foc}
a) $\theta_{\mbox{\small max}}(\eps)$ is analytic in $\eps$ for $\eps\in [0,1)$ and
\begin{align}
\theta_{\mbox{\small max}}^{(n)}(0)&=-\left(\frac{\mathrm{d}^2H(X_0|X_{-1})}{\mathrm{d}\theta^2}\left(\frac{1}{\lambda^2}\right)\right)^{-1}\nonumber\\
&\hspace{-1cm} \quad\times\left\{\sum_{k=1}^{n}{n\choose k}k!\sum_{m_1,m_2,\cdots,m_{n-k}}a(m_1,\cdots,m_{n-k})\right. \quad\frac{\mathrm{d}^{m_1+\cdots+m_{n-k}+1}H(X_0|X_{-k-1})}{\mathrm{d}\theta^{m_1+\cdots+m_{n-k}+1}}\left(\frac{1}{\lambda^2}\right)\prod_{j=1}^{n-k}(\theta_{\mbox{\small max}}^{(j)}(0))^{m_j}\nonumber\\
&\quad+\left.\sum_{m_1,m_2,\cdots,m_{n}=0}a(m_1,\cdots,m_{n})\frac{\mathrm{d}^{m_1+\cdots+m_{n}+1}H(X_0|X_{-1})}{\mathrm{d}\theta^{m_1+\cdots+m_{n}+1}}\left(\frac{1}{\lambda^2}\right)\prod_{j=1}^{n}(\theta_{\mbox{\small max}}^{(j)}(0))^{m_j}\right\}, \label{theta-max}
\end{align}
where $\sum\limits_{m_1,m_2,\cdots,m_{n-k}}$ is taken over all nonnegative intergers $m_1,\cdots,m_{n-k}$ satisfying the constraint
$$m_1+2m_2+\cdots+(n-k)m_{n-k}=n-k$$
and
$$a(m_1,\cdots,m_{n-k})=\frac{(n-k)!}{m_1!1^{m_1}\cdots m_{n-k}!(n-k)^{m_{n-k}}}.$$

b) $C^{(1)}(\mathcal{S}_0, \eps)$ is analytic in $\eps$ for $\eps\in [0,1)$ with the following Taylor series expansion around $\eps=0$:
{\small \begin{equation}  \label{first-order-full}
\hspace{-1.5cm} C^{(1)}(\mathcal{S}_0, \eps) = \sum_{n=0}^{\infty} \left(\left.\frac{\mathrm{d}^nG_0(\eps)}{\mathrm{d}\eps^n}\right|_{\eps=0}+\left.\frac{d^{n-1}(G_1(\eps)-2G_0(\eps))}{\mathrm{d}\eps^{n-1}}\right|_{\eps=0}\\
+\left.\sum_{k=2}^{n}{n\choose k}\frac{\mathrm{d}^{n-k}}{\mathrm{d}\eps^{n-k}}\left\{(G_k(\eps)+G_{k-2}(\eps)-2G_{k-1}(\eps))\right\}\right|_{\eps=0} \right) \eps^n,
\end{equation}}
where $G_k(\eps)=H(X_0|X_{-k-1})(\theta_{\mbox{\small max}}(\eps))$.
\end{theorem}

\begin{proof}
a) For $\eps>0$,
$$
I(X;Y)=\begin{cases}0&\theta=0\ \mbox{or}\ 1,\\>0&\theta\in(0,1).\end{cases}
$$
With Theorem~\ref{concavitybec} establishing the concavity of $I(X; Y)$, $\theta_{\mbox{\small max}}$ should be the unique zero point of the derivative of the mutual information rate. So, $\theta_{\mbox{\small max}}(\eps)\in (0,1)$ and satisfies
\begin{equation}\label{dvanish}
0=\frac{\mathrm{d} I(X;Y)}{\mathrm{d}\theta}=(1-\eps)^2\sum_{k=0}^{\infty}\frac{\mathrm{d}H(X_0|X_{-k-1})}{\mathrm{d}\theta}\eps^{k}.
\end{equation}
According to the analytic implicit function theorem~\cite{stevenimplicitfunctiontheorem}, $\theta_{\mbox{\small max}}(\eps)$ is analytic in $\eps$ for $\eps\in[0,1)$. In the following the $s$-th order derivative of $\theta_{\mbox{\small max}}(\eps)$ at $\eps=0$ is computed. It follows from the Leibniz formula and the Faa di Bruno formula~\cite{faadibruno} that
\begin{align}\label{lfformula}
0&=\left.\sum_{k=0}^{\infty}\frac{\mathrm{d}^n}{\mathrm{d}\eps^n}\left\{\frac{\mathrm{d}H(X_0|X_{-k-1})}{\mathrm{d}\theta}\eps^{k}\right\}\right|_{\eps=0}\nonumber\\
&=\left.\sum_{k=0}^{n}k!{n\choose k}\frac{\mathrm{d}^{(n-k)}}{\mathrm{d}\eps^{n-k}}\left\{\frac{\mathrm{d}H(X_0|X_{-k-1})}{\mathrm{d}\theta}\right\}\right|_{\eps=0}\nonumber\\
&=\sum_{k=0}^{n}k!{n\choose k}\sum_{m_1,m_2,\cdots,m_{n-k}}a(m_1,\cdots,m_{n-k})\frac{\mathrm{d}^{m_1+\cdots+m_{n-k}+1}H(X_0|X_{-k-1})}{\mathrm{d}\theta^{m_1+\cdots+m_{n-k}+1}}\left(\frac{1}{\lambda^2}\right)\prod_{j=1}^{n-k}(\theta_{\mbox{\small max}}^{(j)}(0))^{m_j}
\end{align}
which immediately implies a).

b) Note that
\begin{align*}
C^{(1)}(\mathcal{S}_0, \eps)&=(1-\eps)^{2}\sum_{k=0}^{\infty}G_{k}(\eps)\eps^{k}\\
&=G_0(\eps)+(G_1(\eps)-2G_0(\eps))\eps+\sum_{k=2}^{\infty}(G_k(\eps)+G_{k-2}(\eps)-2G_{k-1}(\eps))\eps^{k}.
\end{align*}
It then follows from the Leibniz formula that
\begin{align*}
{}&\hspace{-1cm} \frac{\mathrm{d}^n}{\mathrm{d}\eps^{n}}\left\{\sum_{k=2}^{\infty}(G_k(\eps)+G_{k-2}(\eps)-2G_{k-1}(\eps))\eps^{k}\right\}\\
&=\sum_{k=2}^{\infty}\sum_{t=0}^{n}{n\choose t}\frac{\mathrm{d}^{n-t}}{\mathrm{d}\eps^{n-t}}\left\{(G_k(\eps)+G_{k-2}(\eps)-2G_{k-1}(\eps))\right\}\eps^{k-t}.
\end{align*}
Therefore,
\begin{align*}
\left.\frac{\mathrm{d}^{n}C^{(1)}(\mathcal{S}_0, \eps)}{\mathrm{d}\eps^s}\right|_{\eps=0}&=\left.\frac{\mathrm{d}^nG_0(\eps)}{d\eps^n}\right|_{\eps=0}+\left.\frac{\mathrm{d}^{n-1}(G_1(\eps)-2G_0(\eps))}{\mathrm{d}\eps^{n-1}}\right|_{\eps=0}\\
&\hspace{4.5mm}+\left.\sum_{k=2}^{\infty}\sum_{t=0}^{n}{n\choose t}\frac{\mathrm{d}^{n-t}}{\mathrm{d}\eps^{n-t}}\left\{(G_k(\eps)+G_{k-2}(\eps)-2G_{k-1}(\eps))\right\}\eps^{k-t}\right|_{\eps=0}\\
&=\left.\frac{\mathrm{d}^nG_0(\eps)}{\mathrm{d}\eps^n}\right|_{\eps=0}+\left.\frac{\mathrm{d}^{n-1}(G_1(\eps)-2G_0(\eps))}{\mathrm{d}\eps^{n-1}}\right|_{\eps=0}\\
&\hspace{4.5mm}+\left.\sum_{k=2}^{n}{n\choose k}\frac{\mathrm{d}^{n-k}}{\mathrm{d}\eps^{n-k}}\left\{(G_k(\eps)+G_{k-2}(\eps)-2G_{k-1}(\eps))\right\}\right|_{\eps=0},
\end{align*}
which immediately implies b).
\end{proof}

Despite their convoluted looks, (\ref{theta-max}) and (\ref{first-order-full}) are explicit and computable. Below, we list the  coefficients of $C^{(1)}(\mathcal{S}_0, \eps)$ (in bits) and $\theta_{\mbox{\small max}}(\eps)$ up to the third order, which are numerically computed according to (\ref{theta-max}) and (\ref{first-order-full}) and rounded off to the ten thousandths decimal digit:

\begin{table}[htbp]
\renewcommand{\arraystretch}{1.3}
\caption{}
\centering
\vspace{0.3cm}
\begin{tabular}{|c|c|c|c|c|}
\hline
&$\eps^0$ &$\eps^1$& $\eps^2$&$\eps^3$\\
\hline
$\theta_{\mbox{max}}(\eps)$&0.3820& 0.0462&0.1586&0.2455\\
\hline
$C_{1}(\eps)$&0.6942&-0.6322& 0.0159& -0.0625\\
\hline
\end{tabular}
\end{table}
\section{Feedback with Input-Constraint} \label{feedback-section}

In this section, we consider the input-constrained erasure channel (\ref{mec}) as in Section~\ref{introduction-section} however with possible feedback, and we are interested in comparing its feedback capacity $C_{FB}(\mathcal{S}, \eps)$ and its non-feedback capacity $C(\mathcal{S}, \eps)$. The following theorem states that for the erasure channel without any input-constraint, feedback does not increase the capacity and both of them can be computed explicitly. This result is in fact implied by Theorem $12$ in~\cite{vwerasure}, where a random coding argument has been employed in the proof; we nonetheless give an alternative proof in Appendix~\ref{FB-NFB} for completeness.
\begin{theorem}\label{fbnot}
For the erasure channel~(\ref{mec}) without any input constraints, feedback does not increase the capacity, and we have
$$
C_{FB}(\mathcal{X}^*, \eps)=C(\mathcal{X}^*, \eps)=(1-\eps)\log K.
$$
\end{theorem}

On the other hand, we will show in the following that feedback may increase the capacity when the input constraint in the erasure channel is non-trivial. As elaborated below, this is achieved by comparing the asymptotics of the feedback capacity and the non-feedback capacity for a special input-constrained erasure channel.

In~\cite{oron2015}, Sabag {\em et al.} computed an explicit formula of feedback capacity for BEC with $(1,\infty)$-RLL input constraint $\mathcal{S}_0$.
\begin{theorem}\label{fbc}~\cite{oron2015}
The feedback capacity of the $(1,\infty)$-RLL input-constrained erasure channel is
$$
C_{FB}(\mathcal{S}_0, \eps)=\max_{0\le p\le\frac{1}{2}}\frac{H(p)}{p+\frac{1}{1-\eps}},
$$
where the unique maximizer $p(\eps)$ satisfies $p=(1-p)^{2-\eps}.$
\end{theorem}

Clearly, the explicit formula in Theorem~\ref{fbc} readily gives the asymptotics of the feedback capacity.

To see this, note that $p(0)={1}/{\lambda^2}$ and $p(1)={1}/{2}$. Straightforward computations yield
$$
\frac{\mathrm{d}\log p(\eps)}{\mathrm{d}\eps}=-\frac{(1-p(\eps))\log p(\eps)}{(1-p(\eps)+p(\eps)(2-\eps))(2-\eps)}.
$$
Hence,
\begin{align*}
C_{FB}(\mathcal{S}_0, \eps)&=-\frac{H(p(\eps))}{p(\eps)+\frac{1}{1-\eps}}\\
&=-\frac{(1-\eps)\log p(\eps)}{{2-\eps}}\\
&=-\frac{1}{2}\log p(\eps)+\frac{1}{2}\log p(\eps)\sum_{k=1}^{\infty}\left(\frac{\eps}{2}\right)^{k}\\
&=-\frac{1}{2}\log p(0)-\frac{1}{2}\left.\frac{\mathrm{d}\log p(\eps)}{\mathrm{d}\eps}\right|_{\eps=0}\eps+\frac{\eps}{4}\log p(0)+O(\eps^2)\\
&=\log \lambda-\frac{\lambda^2}{\lambda^2+1}\log \lambda \cdot \eps+O(\eps^2).
\end{align*}
It then follows from straightforward computations that for the case when $\eps$ is close to $0$, $C(\mathcal{S}_0, \eps)<C_{FB}(\mathcal{S}_0, \eps)$. So, we have proven the following theorem:
\begin{theorem} \label{yonglong-feedback}
For a BEC($\eps$) with the $(1, \infty)$-RLL input constraint, feedback increases the channel capacity when $\eps$ is small enough.
\end{theorem}

\begin{rem}
An independent work in~\cite{th16} also found that feedback does increase the capacity of a BEC($\eps$) with the same input constraint $\mathcal{S}_0$, by comparing a tighter bound of non-feedback capacity $C(\mathcal{S}_0, \eps)$, obtained via a dual capacity approach, with the feedback capacity $C_{FB}(\mathcal{S}_0, \eps)$.
\end{rem}

\begin{rem}
Recently, Sabag {\em et al.}~\cite{oron2015bsc} also computed an explicit asymptotic formula for the feedback capacity of a BSC($\eps$) with the input supported on the $(1,\infty)$-RLL constraint. By comparing the asymptotics of the feedback capacity with the that of non-feedback capacity~\cite{asymptotics-binary}, they showed that feedback does increase the channel capacity in the high SNR regime.

It is well known that for any memoryless channel without any input constraint, feedback does not increase the channel capacity. Theorem~\ref{fbnot} states that when there is no input constraint, the feedback does not increase the capacity of the erasure channel even with the presence of the channel memory. Theorem~\ref{yonglong-feedback} says that feedback may increase the capacity of input-constrained erasure channels even if there is no channel memory. These two theorems, together with the results in~\cite{oron2015bsc, th16}, suggest the intricacy of the interplay between feedback, memory and input constraints.
\end{rem}

\section*{Appendices} \appendix

\section{Proof of Theorem~\ref{fbnot}} \label{FB-NFB}

We first prove that
\begin{equation} \label{C-NFB}
C(\mathcal{X}^*, \eps)=(1-\eps)\log K.
\end{equation}

A similar argument using the independence of $\{X_i\}$ and $\{E_i\}$ as in the proof of~(\ref{pformula}) yields that
$$
p(y_1^n)=P(E_{\mathcal{I}(y_1^n)}=1, E_{\bar{\mathcal{I}}(y_1^n)}=0)P(X_{\mathcal{I}(y_1^n)}=y_{\mathcal{I}(y_1^n)}).
$$
It then follows that
\begin{align}
H(Y_1^n)&=-\sum_{y_1^n}p(y_1^n)\log p(y_1^n)\nonumber\\
&=-\sum_{y_1^n}p(y_1^n)\log P(E_{\mathcal{I}(y_1^n)}=1, E_{\bar{\mathcal{I}}(y_1^n)}=0)-\sum_{y_1^n}P(y_1^n)
\log P(X_{\mathcal{I}(y_1^n)}=y_{\mathcal{I}(y_1^n)})\nonumber\\
&=-\sum_{D \subseteq [1, n]}\sum_{y_1^n:\mathcal{I}(y_1^n)=D}P(E_{D}=1, E_{D^C}=0)P(X_{D}=y_{D})\log P(E_{D}=1, E_{D^c}=0)\nonumber\\
&{}\hspace{4.5mm}-\sum_{D \subseteq [1, n]}\sum_{y_1^n:\mathcal{I}(y_1^n)=D}P(E_{D}=1, E_{D^c}=0)P(X_{D}=y_{D})
\log P(X_{D}=y_{D})\nonumber\\
&=-\sum_{D \subseteq [1, n]}P(E_{D}=1, E_{D^c}=0)\log P(E_{D}=1, E_{D^c}=0)+\sum_{D\subseteq [1, n]}P(E_{D}=1, E_{D^c}=0)H(X_D)\nonumber\\
&=\sum_{D \subseteq [1, n]}P(E_{D}=1, E_{D^c}=0)H(X_D)+H(E_1^n)\nonumber\\
&\le H(E_1^n)+\sum_{D \subseteq [1, n]}P(E_{D}=1, E_{D^c}=0)|D| \log K\nonumber\\
&=H(E_1^n)+\EE[E_1+\cdots+E_n] \log K \label{paomo},
\end{align}
where the only inequality becomes equality if $\{X_n\}$ is i.i.d. with the uniform distribution. It then further follows that
\begin{align*}
C(\mathcal{X}^*, \eps)& =\lim_{n\to \infty}\frac{1}{n}\sup_{p(x_1^n)}I(X_1^n;Y_{1}^n)\\
&=\lim_{n\to \infty}\frac{1}{n}\sup_{p(x_1^n)} (H(Y_1^n)-H(Y_1^n|X_1^n))\\
&=\lim_{n\to \infty}\frac{1}{n}\sup_{p(x_1^n)}(H(Y_1^n)-H(E_1^n))\\
&\stackrel{(a)}{\leq} \lim_{n\to \infty}\frac{1}{n} \EE[E_1+\cdots+E_n] \log K\\
&\stackrel{(b)}{=} P(E_1=1) \log K\\
&= (1-\eps) \log K ,
\end{align*}
where $(a)$ follows from (\ref{paomo}) and $(b)$ follows from the ergodicity of $\{E_n\}$. The desired (\ref{C-NFB}) then follows from the fact that the only inequality $(a)$ becomes equality if $\{X_n\}$ is i.i.d. with the uniform distribution.

We next prove that
$$
C_{FB}(\mathcal{X}^*, \eps) \leq (1-\eps)\log K,
$$
which, together with (\ref{C-NFB}), immediately implies the theorem.

Let $W$, independent of $\{E_i\}$, be the message to be sent and $X_{i}(W,Y_{1}^{i-1})$ denote the encoding function. As shown in~\cite{tatikondafb},
$$C_{FB}(\mathcal{X}^*, \eps)=\lim_{n\to \infty}\frac{1}{n}\sup_{\{p(X_{i}=\cdot|X_1^{i-1},Y_{1}^{i-1}):i=1,\cdots,n\}}I(W;Y_1^{n}).$$
Using the chain rule for entropy, we have
\begin{align}
H(Y_{1}^{n}|W)&=\sum_{i=1}^{n}H(Y_i|W, Y_{1}^{i-1})\nonumber\\
&\stackrel{(a)}{=}\sum_{i=1}^{n}H(E_iX_i|W, Y_{1}^{i-1}, X_1^{i},E_1^{i-1}) \nonumber\\
&=\sum_{i=1}^{n}H(E_i|W, Y_{1}^{i-1}, X_1^{i}, E_1^{i-1}) \nonumber\\
&\stackrel{(b)}{=}\sum_{i=1}^{n}H(E_i|E_1^{i-1}) \nonumber \\
&=H(E_1^n),\label{HZW}
\end{align}
where (a) follows from the fact that $X_i$ is a function of $W$ and $Y_{1}^{i-1}$ and $E_i=0$ if and only if $Y_i=0$, (b) follows from the independence of $W$ and $\{E_i\}$.

Note that for $y_i\not=0$,
\begin{align}\label{not0p}
p(y_i|w,y_{1}^{i-1})&=P(X_i(w,y_1^{i-1})=y_i, E_i=1|w,y_{1}^{i-1})\nonumber\\
&=P(X_i(w,y_1^{i-1})=y_i|w,y_{1}^{i-1})P(E_i=1|X_i(w,y_1^{i-1})=y_i, w,y_{1}^{i-1})\nonumber\\
&=P(X_i(w,y_1^{i-1})=y_i|w,y_{1}^{i-1})P(E_i=1|E_{\mathcal{I}(y_1^{i-1})}=1, E_{\bar{\mathcal{I}}(y_1^{i-1})}=0).
\end{align}
And for $y_i=0$,
\begin{align}\label{0p}
p(y_i|w,y_{1}^{i-1})&=P(E_{i}=0|E_{\mathcal{I}(y_1^{i-1})}=1, E_{\bar{\mathcal{I}}(y_1^{i-1})}=0,w, y_{1}^{i-1})\nonumber\\
&\stackrel{(a)}{=}P(E_{i}=0|E_{\mathcal{I}(y_1^{i-1})}=1, E_{\bar{\mathcal{I}}(y_1^{i-1})}=0),
\end{align}
where $(a)$ follows from the independence of $W$ and $\{E_i\}$. It then follows that
\begin{align*}
p(y_1^{n})&=\sum_{w}p(w)p(y_1^{n}|w)\\
&=\sum_{w}p(w)\prod_{i=1}^{n}p(y_i|w,y_{1}^{i-1})\\
&\stackrel{(a)}{=}\sum_{w}p(w)\prod_{i\in \mathcal{I}(y_1^{n})}p(y_i|w,y_{1}^{i-1})\prod_{i\in \bar{\mathcal{I}}(y_{1}^n)}p(y_i|w,y_{1}^{i-1})\\
&\stackrel{(b)}{=}\sum_{w}p(w)\prod_{i\in \mathcal{I}(y_1^{n})}p(y_i|w,y_{1}^{i-1})\prod_{i\in \bar{\mathcal{I}}(y_{1}^n)}P(E_i=0|E_{\mathcal{I}(y_1^{i-1})}=1, E_{\bar{\mathcal{I}}(y_1^{i-1})}=0)\\
&=\left\{\sum_{w}p(w)\prod_{i\in \mathcal{I}(y_1^{n})}P(X_i(w,y_1^{i-1})=y_i|w,y_{1}^{i-1})\right\}P(E_{\mathcal{I}(y_1^{n})}=1, E_{\bar{\mathcal{I}}(y_1^{n})}=0),
\end{align*}
where $(a)$ follows from~(\ref{not0p}) and $(b)$ follows from~(\ref{0p}). Since for any $D \subseteq [1, n]$,  $$
\sum_{y_1^n:\mathcal{I}(y_1^n)=D}p(y_1^{n})=P(E_{D}=1, E_{D^c}=0),
$$
which implies that
$$
\left\{q(\tilde{y}_1^n))\triangleq\sum_{w}p(w)\prod_{i\in \mathcal{I}(\tilde{y}_1^{n})}P(X_i(w,\tilde{y}_1^{i-1})=\tilde{y}_i|w,\tilde{y}_{1}^{i-1}):\mathcal{I}(\tilde{y}_1^{n})=\mathcal{I}(y_1^{n})\right\}
$$
is an $M^{|\mathcal{I}({y}_1^{n})|}$-dimensional probability mass function. Therefore, through a similar argument as before, we have
\begin{align}
H(Y_1^n)&=-\sum_{y_1^n}p(y_1^n)\log p(y_1^n)\nonumber\\
&=-\sum_{y_1^n}p(y_1^n)\log P(E_{\mathcal{I}(y_1^n)}=1, E_{\bar{\mathcal{I}}(y_1^n)}=0)-\sum_{y_1^n}P(E_{\mathcal{I}(y_1^n)}=1, E_{\bar{\mathcal{I}}(y_1^n)}=0)q(y_1^n)\log q(y_1^n)\nonumber\\
&=H(E_1^n)-\sum_{D \subseteq [1, n]}\sum_{y_1^n:\mathcal{I}(y_1^n)=D}P(E_{D}=1, E_{D^c}=0)q(y_1^n)
\log q(y_1^n)\nonumber\\
&\le H(E_1^n)+\sum_{D \subseteq [1, n]}P(E_{D}=1, E_{D^c}=0)|D| \log K\nonumber\\
&=H(E_1^n)+\EE[E_1+\cdots+E_n] \log K, \label{HZ}
\end{align}
where the inequality follows from the fact that $q(y_1^n)$ is an $M^{|D|}$-dimensional probability mass function.

Combining~(\ref{HZW}) and~(\ref{HZ}), we have
\begin{align*}
C_{FB}(\mathcal{X}^*, \eps) &=\lim_{n\to \infty}\frac{1}{n}\sup_{\{p(X_{i}=\cdot|X_1^{i-1},y_{1}^{i-1}):i=1,\cdots,n\}}I(W;Y_1^{n})\\
&\le \lim_{n\to \infty}\frac{1}{n} \EE[E_1+\cdots+E_n] \log K\\
&= P(E_1=1)\log K\\
&= (1-\eps) \log K,
\end{align*}
as desired.

\bigskip

{\bf Acknowledgement.} We would like to thank Navin Kashyap, Haim Permuter, Oron Sabag and Wenyi Zhang for insightful discussions and suggestions and for pointing out relevant references that result in great improvements in many aspects of this work.

\end{document}